\documentclass[journal,draftclsnofoot,onecolumn,12pt]{IEEEtran}
\newif\ifonecolumn
\onecolumntrue 
\usepackage{adjustbox}
\usepackage{xcolor}
\usepackage{cite}
\usepackage{amsmath,amssymb,amsfonts}
\usepackage{graphicx,color}
\usepackage{acro}
\usepackage{comment}

\DeclareAcronym{snr}{
  short = SNR,
  long = signal-to-noise ratio,}

  \DeclareAcronym{sar}{
  short = SAR,
  long = synthetic aperture radar,}

\DeclareAcronym{insar}{
  short = InSAR,
  long = interferometric synthetic aperture radar,}
\DeclareAcronym{minlp}{
	short = {MINLP},
	long = {mixed-integer non-linear program},
	long-plural-form = {mixed-integer non-linear programs}
}

\DeclareAcronym{uav}{
        short = {UAV},
        long = {unmanned aerial vehicle},
        long-plural-form = {unmanned aerial vehicles}
}

\DeclareAcronym{fdma}{
	short = {FDMA},
	long = {frequency-division multiple-access},
}

\DeclareAcronym{3d}{
        short = {3D},
        long = {three-dimensional},
}

\DeclareAcronym{2d}{
        short = {2D},
        long = {two-dimensional},
}

\DeclareAcronym{dem}{
        short = {DEM},
        long = {digital elevation model},
}
\DeclareAcronym{gs}{
        short = {GS},
        long = {ground station},
        long-plural-form = {ground stations}
}

\DeclareAcronym{los}{
        short = {LOS},
        long = {line-of-sight},
}

\DeclareAcronym{sca}{
        short = {SCA},
        long = {successive convex approximation},
}

\DeclareAcronym{ao}{
        short = {AO},
        long = {alternating optimization},
}

\DeclareAcronym{aoi}{
	short = {AoI},
	long = {area of interest},
}

\DeclareAcronym{nesz}{
        short = {NESZ},
        long = {noise equivalent sigma zero},
}

\DeclareAcronym{wrt}{
        short = {w.r.t.},
        long = {with respect to},
}
\DeclareAcronym{rhs}{
        short = {r.h.s.},
        long = {right-hand side},
}

\DeclareAcronym{lhs}{
        short = {l.h.s.},
        long = {left-hand side},
}
\DeclareAcronym{hoa}{
        short = {HoA},
        long = {height of ambiguity},
}

\usepackage{textcomp}
\usepackage{soul}
\usepackage{cite}
\usepackage{amsmath,amssymb,amsfonts}
\usepackage{bm}
\usepackage{float}
\usepackage{graphicx}
\usepackage{algorithm}
\usepackage{makecell}
\usepackage{amsthm}
\usepackage{algpseudocode}
\usepackage{tablefootnote}
\newcommand\scalemath[2]{\scalebox{#1}{\mbox{\ensuremath{\displaystyle #2}}}}

\DeclareMathOperator*{\argmax}{arg\,max}
\DeclareMathOperator*{\argmin}{arg\,min}

\DeclareMathOperator*{\erfc}{erfc}
\newtheorem{theorem}{Theorem}[]
\newtheorem{proposition}[]{Proposition}

\newtheorem{lem}[]{Lemma}

\theoremstyle{remark}
\newtheorem{remark}[]{Remark}
\def\BibTeX{{\rm B\kern-.05em{\sc i\kern-.025em b}\kern-.08em
    T\kern-.1667em\lower.7ex\hbox{E}\kern-.125emX}}
\AtBeginDocument{\definecolor{ojcolor}{cmyk}{0.93,0.59,0.15,0.02}}
\def\OJlogo{\vspace{-14pt}\includegraphics[height=20pt]{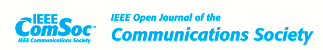}}
\allowdisplaybreaks

\newcolumntype{?}{!{\vrule width 1pt}}
\usepackage[normalem]{ulem}
\newcommand{\stkout}[1]{\ifmmode\text{\sout{\ensuremath{#1}}}\else\sout{#1}\fi}
\begin{document}
\ifonecolumn
\title{Robust Trajectory and Resource Optimization for Communication-assisted UAV SAR Sensing\\
	\thanks{Mohamed-Amine~Lahmeri, Walid R. Ghanem, and Robert Schober are with the  Institute for Digital Communications, Friedrich-Alexander University Erlangen-N\"urnberg (FAU), Germany (emails:\{amine.lahmeri, walid.ghanem, robert.schober\}@fau.de). Christina Bonfert is with the Institute of Microwave Engineering, Ulm, Germany (email: christina.knill@uni-ulm.de). This paper was presented in part at the IEEE Global Communications
		Conference, Rio De Janeiro, Brazil, Dec. 2022  \cite{14}. This work was supported in part by the Deutsche Forschungsgemeinschaft (DFG, German Research Foundation) GRK 2680 – Project-ID 437847244.}
}
\author{\IEEEauthorblockN{Mohamed-Amine~Lahmeri, Walid R. Ghanem, Christina Bonfert, and Robert Schober}\vspace{-15mm}}
\maketitle
\fi
\begin{abstract}
In this paper, we investigate joint 3-dimensional (3D) trajectory planning and resource allocation for rotary-wing unmanned aerial vehicle (UAV) synthetic aperture radar (SAR) sensing. To support emerging real-time SAR applications and enable live mission control, we incorporate real-time communication with a ground station (GS).  The UAV's main mission is the mapping of large areas of interest (AoIs) using an onboard SAR system and transferring the unprocessed raw radar data to the ground in real time. We propose a robust trajectory and resource allocation design that takes into account random UAV trajectory deviations. To this end, we model the UAV trajectory deviations and study their effect on the radar coverage. Then, we formulate a robust non-convex mixed-integer non-linear program (MINLP) such that the UAV 3D trajectory and resources are jointly optimized for maximization of the radar ground coverage. A low-complexity sub-optimal solution for the formulated problem is presented. Furthermore, to assess the performance of the sub-optimal algorithm, we derive an upper bound on the optimal solution based on monotonic optimization theory. Simulation results show that the proposed sub-optimal algorithm achieves close-to-optimal performance and not only outperforms several benchmark schemes but is also robust with respect to UAV trajectory deviations.
\end{abstract}
\begin{IEEEkeywords}
Optimization, communication, sensing, unmanned aerial vehicles, synthetic aperture radar, successive convex approximations, monotonic optimization.
\end{IEEEkeywords}

\maketitle

\section{INTRODUCTION}
\IEEEPARstart{S}{ynthetic }aperture radar (\acs{sar}) is a moving radar system that uses the flight path of a moving platform to simulate a large aperture \cite{14,tutorial}. In classical radar systems,  the spatial resolution depends on the ratio of the wavelength of the radar signal to the physical length of the antenna. This means that spatial resolution can be improved by using large antenna arrays, which are however often not practical. \ac{sar} overcomes this hurdle by using the motion of the moving platform to form a synthetic aperture where observations in different time slots, corresponding to different positions of the trajectory, are jointly processed to create the final image of the target area \cite{8,FFBP}. \ac{sar} can provide high-resolution \ac{2d} and even \ac{3d} images \cite{glacier,tomography}. The benefits of \ac{sar} also include its resilience to harsh weather conditions and its ability to operate at day and night. In addition, \ac{sar} is  non-invasive, which is useful for applications where penetration or direct contact with the \ac{aoi} is not possible or not desirable, such as imaging of glaciers and wetlands \cite{glacier}. Thanks to these features, \ac{sar} sensing has been employed in a  variety of applications, such as environmental monitoring, agriculture, and surveillance.

\ac{sar} sensing has traditionally relied on airborne and spaceborne systems such as satellites and aircraft. However, these types of platforms lack flexibility and their revisit time for a given \ac{aoi} is often large \cite{tutorial}. In this context,  the recent growth of the drone industry has motivated the use of \ac{sar} onboard small and lightweight \acp{uav}, which is promising and challenging at the same time. On the one hand, \ac{sar} sensing benefits from the high flexibility and the easy deployment of small drones at relatively short distances from \acp{aoi}. Consequently, \ac{uav}-based \ac{sar} sensing, which is now possible at low cost, facilitates short revisit times and enables various new applications, such as the monitoring of highly dynamic environments. On the other hand, the drones' limited hardware capacity is confronted with highly computationally complex \ac{sar} image processing algorithms\cite{complexityBP}. This makes real-time \ac{sar}, which is an emerging \ac{sar} application \cite{realtime1,realtime2}, difficult to realize onboard the flying platform. {  In this regard, one potential solution is to transmit the \ac{sar} raw data to the ground for processing, which requires a reliable communication link between the \ac{uav} and the ground \cite{UAVWireless,amine1,open1}. In real-time \ac{sar}, the radar raw data received for each sensing pulse is processed with minimum delay such that the \ac{sar} image frame can be generated as a function of the received \ac{sar} raw data. This facilitates  immediate adjustments based on the processed \ac{sar} data.} Moreover, \ac{uav} instability presents another challenge facing \ac{uav}-\ac{sar} systems. In fact, lightweight drones are vulnerable to external perturbations such as wind and atmospheric turbulence \cite{jitter}. In this context, robust trajectory and resource allocation design \ac{wrt} \ac{uav} jittering were investigated for \ac{uav}-based communication systems \cite{jitter2,robust1}. Another challenge is the drones' limited battery capacity \cite{table2}. Optimizing the \ac{uav}'s energy consumption requires a careful trajectory design as energy consumption and  \ac{sar} sensing performance are tightly coupled.\par 
 { In general, trajectory design and resource allocation for \acp{uav} have been exhaustively studied in the literature \cite{pathUAV,response4} and different optimization techniques ranging from classical optimization \cite{response3} to learning-based methods \cite{open2,response5} have been used to tackle the problem. However, for the \ac{uav}-\ac{sar} use case, trajectory design and resource allocation have not been comprehensively addressed}. In fact, existing experimental works on \ac{uav}-\ac{sar} systems only investigate the practical deployment of the \ac{sar} sensors onboard small drones and report results of  experimental measurement campaigns with intuitive pre-planned flight paths with fixed velocity and altitude \cite{t2,t5,t11}.  Furthermore, existing experimental research works do not account for real-time \ac{sar} applications and are limited to offline image processing without considering a potential wireless communication link to the ground. Among the few works targeting the optimization of \ac{uav}-\ac{sar} systems, the authors in \cite{9} jointly optimized a focus measure for the \ac{sar} images and the estimation of the platform trajectory. However, while the estimation error of the actual trajectory was minimized based on a nominal trajectory, the trajectory itself was not optimized in \cite{9}, and communication with the ground was not considered. Furthermore, for a bistatic \ac{uav}-\ac{sar} system, the authors in \cite{99} investigated  trajectory optimization for a cellular-connected bistatic \ac{uav}-\ac{sar} system, where the \ac{aoi} was illuminated by a ground base station, and the energy consumption was minimized. Yet, the proposed solution is limited to bistatic systems with a stationary ground transmitter and a moving receiver, making it inapplicable for active monostatic systems, where both transmitter and receiver are mounted on the same moving platform.  The authors in \cite{10,11} investigated the use of a  geosynchronous \ac{sar} system, where the \ac{aoi}  was illuminated by a satellite and echoes were collected by a passive sensing \ac{uav}. For instance, in \cite{10}, a multi-objective optimization problem for minimization of the distance traveled by the \ac{uav}-\ac{sar} was formulated and solved. However, the trajectory design of passive bistatic \ac{uav}-\ac{sar} systems is fundamentally different from that of active monostatic \ac{uav}-\ac{sar} systems.\par  
In this paper, we investigate the joint \ac{3d} trajectory and resource allocation design for communication-assisted \ac{uav}-based \ac{sar} sensing. Different from other works, we focus on active monostatic \ac{sar} sensing, where the \ac{uav}-\ac{sar} enables real-time \ac{sar} applications by transmitting the collected radar data  instantaneously to the ground for processing. To this end, the \ac{uav} \ac{3d} trajectory and resources are jointly optimized to maximize the radar ground coverage under communication, radar, and energy constraints.  Our main contributions are summarized as follows: 
\begin{itemize} 
	\item The joint \ac{3d} trajectory and resource allocation algorithm design is formulated as a non-convex \ac{minlp}. Our problem formulation takes into account  the constraints imposed by both the radar and communication subsystems.
	\item The unavoidable \ac{uav} trajectory deviations are modeled based on their statistics and their effect on the \ac{uav}-\ac{sar} ground radar coverage is unveiled. 
	\item Exploiting successive convex approximation (SCA), we propose a low-complexity sub-optimal solution for the formulated problem, which is robust  \ac{wrt} \ac{uav} trajectory deviations. Furthermore, based  on monotonic optimization theory, we provide an upper bound on the optimal solution to assess the performance of the proposed sub-optimal algorithm. 
	\item Our simulation results show that the proposed sub-optimal algorithm achieves close-to-optimal performance. In addition, we demonstrate the superiority of the proposed solution in terms of radar coverage compared to several benchmark schemes and confirm its robustness against random trajectory deviations.
\end{itemize}
We note that this article extends the corresponding conference version \cite{14}. In \cite{14}, only a non-robust sub-optimal algorithm was provided, and random \ac{uav} trajectory deviations were not accounted for. Moreover, the  gap  between the sub-optimal scheme and the optimal solution was not assessed in \cite{14}. \par 
The remainder of this article is organized as follows. In Section \ref{Sec:SystemModel}, we present the system model for the considered \ac{uav}-\ac{sar} system. In Section \ref{Section:FlightDeviation}, we model random \ac{uav} trajectory deviations and propose a concept for robust trajectory design. In Section  \ref{Sec:ResourceAllocation}, we first formulate an optimization problem for the joint \ac{3d} trajectory and resource allocation design for \ac{uav}-\ac{sar} systems. Then, a robust sub-optimal solution and a corresponding upper bound are provided. In Section \ref{Sec:SimulationResults}, we evaluate the performance of the proposed scheme using the derived upper bound and other benchmark schemes. Finally, Section \ref{Sec:Conclusion} concludes the paper { and provides some directions for future research}.\par
{\em Notations}: In this paper, lower-case letters $x$ refer to scalar numbers, boldface lower-case letters $\mathbf{x}$ denote vectors, and boldface upper-case letters $\mathbf{X}$ refer to matrices.  $\{a, ..., b\}$ denotes the set of all integers between $a$ and $b$. $|\cdot|$ denotes the absolute value operator, whereas $||\cdot||_2$ refers to the $L_2$-norm. Operators $[x]^+$ and $\lceil x \rceil$ correspond to $\max(0,x)$ and the ceiling function of $x$, respectively. For given sets $\mathcal{X}_1$ and $\mathcal{X}_2$, $\mathcal{X}_1 \setminus \mathcal{X}_2$  is the set that comprises all elements of $\mathcal{X}_1$ that are not in $\mathcal{X}_2$.  $\mathbb{N}$ represents the set of natural numbers. For arbitrary integer $k$, $\mathbb{N}_k$ is the set $\{1,..,k\}$. $\mathbb{R}^{N \times 1 }$ represents the set of $N \times 1$ vectors with real-valued entries, $\mathbb{R}^{N\times1 }_+$ denotes the set of $N \times 1$ vectors with non-negative real entries, and  $\mathbb{R}^{N \times M }$ denotes the set of real $N\times M $ matrices. For $a, b \in \mathbb{N}$, $a$ mod $b$ is the remainder of the division of $a$ by $b$.  For vector $\mathbf{x}$, vector $\mathbf{x}^T$ stands for its transpose. For vectors $\mathbf{x} = [x(1), ...,x(N)]^T \in \mathbb{R}^{N \times 1} $ and $\mathbf{y} = [y(1), ...,y(N)]^T \in \mathbb{R}^{N \times 1}$,  $\mathbf{x}\leq \mathbf{y}$ and $\mathbf{x} <\mathbf{y}$ define element-wise order relations on vectors $\mathbf{x}$ and $\mathbf{y}$, i.e., $x(n) \leq y(n)$  and $x(n) < y(n)$, $  \forall n$, respectively.   For $ \mathbf{A} \in \mathbb{R}^{N\times N}$, $\mathbf{A}\succeq \mathbf{0} $ indicates that $\mathbf{A}$ is a positive semi-definite matrix.  For a given random variable $X \in \mathbb{R}$, $F_X $ denotes its cumulative distribution function (CDF) and $X \sim \mathcal{N}(\mu,\sigma)$ means that $X$ is Gaussian distributed with mean $\mu$ and standard deviation $\sigma$. $\mathbb{P}(\cdot)$ denotes the probability operator. For a real-valued function $f(x)$, $f'(x)$ denotes the derivative of $f$. $\rm erf(\cdot)$ is the Gaussian error function. {In Table \ref{Tab:System0}, we introduce the definition of most important variables used in this paper.}
\ifonecolumn
\begin{table}[H]
	\centering
	\caption{Definition of most important variables.}
	\begin{adjustbox}{width=\columnwidth,center}
		\begin{tabular}{?c|c?c|c?}
			\Xhline{3\arrayrulewidth}
			Variable              & 	Definition  & Variable & Definition   \\ \Xhline{3\arrayrulewidth}
			$L$    & Length of the \ac{aoi} & 	$\mathrm{SNR}_n(\mathbf{P}_{\mathrm{sar}},\mathbf{u})$  &  Instantaneous radar \ac{snr}  \\ \hline
			$N$  & Number of azimuth sweeps  &  $R_{\rm sl}$ &Synchronization and localization data size \\ \hline
			$M$         & Number of time slots per azimuth scan &  $R_{\mathrm{min},n}(\mathbf{u})$  & Instantaneous \ac{sar} data rate   \\ \hline
			$z_{\rm max}$ and  	$z_{\rm min}$  & \ac{uav} maximum and minimum altitudes  & $R_n(\mathbf{P}_{\mathrm{com}},\mathbf{u})$  & Instantaneous \ac{uav} data rate\\ \hline 
			$\mathbf{u}(n)$  & \ac{uav} \ac{3d} location in time slot $n$ &$B_c$         & Communication bandwidth      \\ \hline   
			$q(n)$ & Remaining battery capacity in time slot $n$  &$\gamma$ &Reference channel power gain\\  \hline  
			$P_{\rm sar}(n)$  &  Instantaneous \ac{sar} transmit power    & $\mathbf{g}$  & \Ac{gs} \ac{3d} location   \\ \hline 	
			$P_{\rm com}(n)$ &  Instantaneous communication transmit power  &  $d_n(\mathbf{u})$  & \ac{uav}-\ac{gs} distance in time slot $n$  \\ \hline  
			$P_{\rm sar}^{\rm max}$ &  Maximum \ac{sar}  transmit power &  $v$    & \ac{uav} velocity   \\\hline   
			$P_{\rm com}^{\rm max}$&  Maximum communication transmit power & $q_{\rm start}$      & Initial \ac{uav} battery capacity   \\ \hline
			$\mathbf{P}_{\rm sar}$ &  \ac{sar}  transmit power vector & $P_{\rm prop}$  & \ac{uav}  propulsion power    \\\hline   
			$\mathbf{P}_{\rm com}$&  Communication transmit power vector&$E_{n}(\mathbf{P}_{\mathrm{sar}}, \mathbf{P}_{\mathrm{com}})$ & Total energy consumed during time slot $n$ \\ \hline  
			$\tau_p$ and $\mathrm{PRF}$        & Pulse duration and repetition frequency &$S_n(\mathbf{u})$ & \ac{uav} \ac{sar} swath width in time slot $n$\\  \hline 
			$\theta_d$ and $\theta_{\rm 3db}$  & 	Radar depression angle and beamwidth&$	C(\mathbf{u})$ & Total \ac{sar} ground coverage \\   \hline 
			$B_r$         &  Radar bandwidth& $r$  & Coverage reliability level     \\  \hline 	
			$f$  & Radar center frequency &$\boldsymbol{\delta}_{\mathbf{p}_N}^{r}$ & Near-range compensation vector\\\hline
			 $\mathrm{SNR_{min}}$& Minimum radar \ac{snr} & $\boldsymbol{\delta}_{\mathbf{p}_F}^{r}$ & Far-range compensation vector \\\hline\Xhline{3\arrayrulewidth}
		\end{tabular}
	\end{adjustbox}	
	\label{Tab:System0}
\end{table}
\fi

\section{System Model} \label{Sec:SystemModel}

\ifonecolumn
\begin{figure}[]
	\centering
	\includegraphics[width=4in]{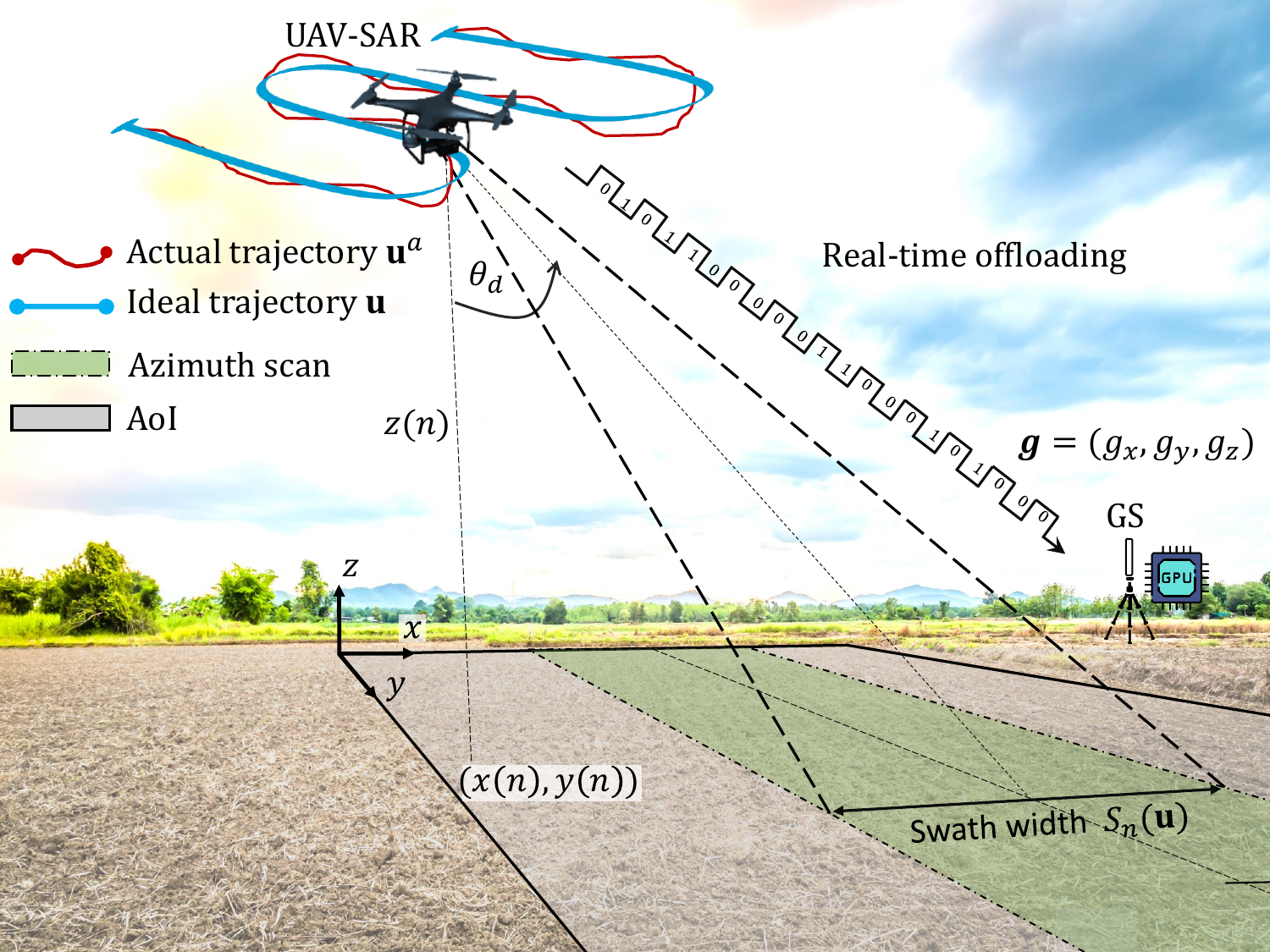}
	\caption{System model for \ac{uav}-\ac{sar} system with real-time data offloading to a GS.}
	\label{fig:SystemModel}
\end{figure} 
\fi
We consider a rotary-wing \ac{uav} equipped with a side-looking radar system designed to perform \ac{sar}. The main mission of the \ac{uav} is to scan a certain \ac{aoi} on the ground. To enable live mission control of the \ac{uav}-\ac{sar} system, real-time \ac{sar} processing is facilitated by transmitting the sensing data to a \ac{gs} in real-time.
\subsection{Trajectory Design}
\label{Sec:TrajectoryDesign}
As a stripmap \ac{sar} mode \cite{tutorial} is employed, we assume that the \ac{aoi} is rectangular with a given length $L$. Without loss of generality, we place the origin of the adopted coordinate system at the corner of the \ac{aoi}, where the \ac{uav} mission starts, see Figure \ref{fig:SystemModel}.  To enable real-time processing of the \ac{sar} images, we assume that frequency-domain algorithms are employed due to their computational efficiency\cite{FFBP, realtime1,realtime2}. { These algorithms impose the use of a linear \ac{uav} trajectory with fixed velocity  $v >0 $ and zero acceleration\cite{realtime3}.} The $x$-axis and $y$-axis define the range and azimuth direction, respectively, see Figure \ref{fig:SystemModel}. 
We discretize time such that the length of the area scanned in $y$-direction in time slot $n$ of duration $\delta_t$, denoted by $\Delta_s$, is given by:
\begin{equation}
	\Delta_{s} =\delta_t v .
\end{equation}
The \ac{uav} trajectory consists of $N$ linear azimuth scans, where each scan corresponds to a traversal in the $y$-direction from one edge of the \ac{aoi} (i.e., $y=0$) to the other  (i.e., $y=L$) and vice versa, see Figure \ref{fig:SystemModel}.   Here, the use of a linear trajectory is imposed by the \ac{sar} stripmap mode \cite{tutorial}. In each azimuth scan, $M$ uniform pulses are transmitted from the radar system\footnote{Here, $\Delta_s$ is set such that $M\Delta_s=L$.}, making the total number of time slots equal to $NM$. The drone location in time slot $n  \in \mathbb{N}_{NM}$ is defined by $\mathbf{u}(n)=[x(n),y(n),z(n)]^T$, where $x(n)$ and $y(n)$ represent the position in the $xy$-plane, while $z(n)$ denotes the drone altitude. The \ac{uav} trajectory $\mathbf{u}=[\mathbf{u}(1),...,\mathbf{u}(NM)]$ is designed such that successive azimuth scans are  ideally adjacent with no gaps in the coverage. At the end of each azimuth scan,  the \ac{uav} makes its turns and moves to the appropriate $x$-position for the next scan as defined later in (\ref{eq:turn1}) and (\ref{eq:turn2}).  The \ac{uav} altitude can be adjusted from scan to scan, but must be fixed during a given azimuth scan to have a fixed width of the on-ground footprint of the sensing beam along the range direction{, i.e., fixed swath width}. Consequently, the radar power {is also adjusted from one azimuth scan to the next, such that the minimum required sensing \ac{snr}  is achieved, but is fixed during a given scan }. To account for these limitations of practical radar systems,  we define two sets of time slots, denoted by $\mathcal{A}$ and $\mathcal{A}^c$, as follows:\begin{align}
	&\mathcal{A}=\{ k\in \mathbb{N}_{NM} |  (k-1) \;\mathrm{mod} \; M \neq 0 \},\\
	\label{eq:setA}
	&\mathcal{A}^c=\mathbb{N}_{NM}\setminus \mathcal{A}=\{1+(k-1)M, k \in \mathbb{N}_N\}.
\end{align}

Sets $\mathcal{A}$ and $\mathcal{A}^c$ will be used to (i) force a back-and-forth linear motion and (ii) set all of the relevant radar parameters, such as the radar power and altitude, constant during a given scan \cite{13,t5}. To this end, set $\mathcal{A}^c$ contains the $N$ indices of the first time slots at the beginning of each azimuth scan, where the radar parameters can be adjusted, whereas set $\mathcal{A}$ contains the indices of all other time slots, where the parameters are fixed.
The following constraints on the drone trajectory are imposed: 
\ifonecolumn
\begin{align}
	&\mathrm{C1}:x(1)=-c_1\;z(1),\label{eq:turn1}\\
	&\mathrm{C2}: x(n)=x(n-1)+c_2\;z(n-1)-c_1\;z(n), \forall n \in \mathcal{A}^c\setminus\{1\},\label{eq:turn2}\\
	&\mathrm{C3}: x(n+1)=x(n), \forall n \in \mathcal{A},\\
	&\mathrm{C4}: z(n+1)=z(n), \forall n \in \mathcal{A},\label{eq:turn4}
\end{align}
\fi

\ifonecolumn
\begin{figure}[]
	\centering
	\includegraphics[width=3.5in]{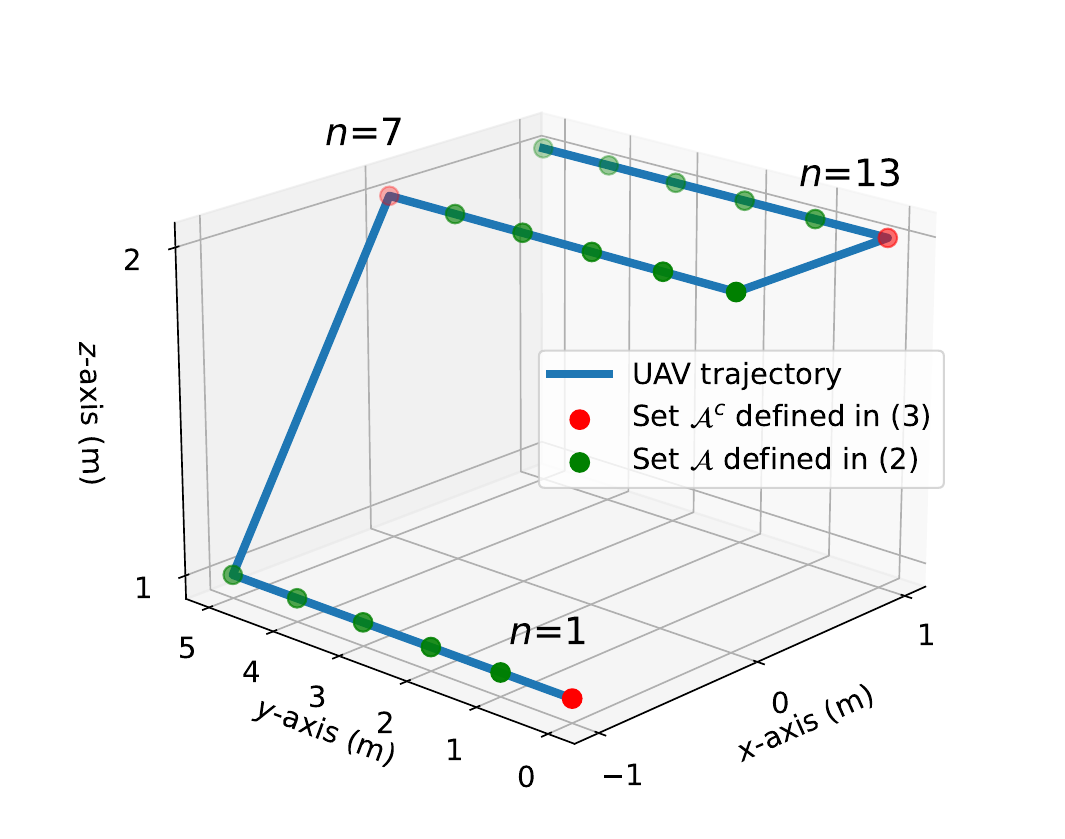}
	\caption{Illustration of a Boustrophedon-shaped trajectory as imposed by constraints $\mathrm{C1-C4}$ for $N=3$, $M=6$, and a set of altitudes $\{1,2,2\}$ m.}
	\label{fig:C1_4} 
\end{figure}
\fi
where $c_1=\tan(\theta_1)$, $c_2=\tan(\theta_2)$, $\theta_1=\theta_d-\frac{\Theta_{3\mathrm{dB}}}{2}$,  $\theta_2=\theta_d+\frac{\Theta_{3\mathrm{dB}}}{2}$,  $\Theta_{3\mathrm{dB}}$ is the radar antenna 3dB  beam width, and $\theta_d$ is the radar depression angle, see Figure \ref{fig:SystemModel}.
Constraint $\mathrm{C2}$ ensures that two successive radar azimuth scans of the ground are adjacent without any gap between them. During each scan, constraint $\mathrm{C3}$ imposes a linear motion, while constraint $\mathrm{C4}$ ensures a fixed altitude during a scan.  { In a nutshell, constraints $\mathrm{C1-C4}$ give rise to a Boustrophedon-shaped trajectory, as illustrated in Figure \ref{fig:C1_4}.} Forcing a linear motion for stripmap \ac{uav}-\ac{sar} does not mean that the trajectory is completely fixed. In fact, in each turn, the drone's $z$-position (altitude) and $x$-position can be adjusted. Only the $y$-position is fixed such that back-and-forth stripmap mode motion is guaranteed. The $y$-position is defined as follows: 
\begin{equation}
	y(1)=0,  \; y(n)=y(n-1) + c(n-1) \Delta_s, \forall n \in \mathbb{N}_{NM}\setminus \{1\},
\end{equation}where vector $\mathbf{c} =[c(1),..., c(NM) ]^T \in \mathbb{R}^{NM\times 1}$ is defined as follows: 
\begin{equation}
	\mathbf{c}= [\overbrace{\underbrace{1,...,1}_\text{$M$ times},\underbrace{-1,...,-1,}_\text{$M$ times}...]^T}^\text{$N\times M$ times}.
\end{equation}
\subsection{Radar Imaging}
Since we operate at relatively low altitudes, the maximal imageable area on the ground is limited by the antenna beam width. In time slot $n$, the width in $x$-direction of the on-ground footprint of the sensing beam, also referred to as \ac{sar} swath width, is denoted by $S_n$ and given as follows: 
\begin{equation}
	S_n(\mathbf{u})= \left(c_2-c_1\right)z(n).
\end{equation}
Following  the trajectory $\mathbf{u}$, {the total ground area covered by the drone is the sum of all infinitesimal small swath widths, and therefore}  can be  approximated\footnote{The approximation is due to the elliptical shape of the beam footprint on the ground and becomes negligible for large $M$. }  as follows: 
\begin{equation} \label{eq:covergeMetric}
	C(\mathbf{u})=\sum_{n=1}^{NM} \Delta_s S_n(\mathbf{u}).
\end{equation}
The data rate produced by the \ac{sar} sensing data in time slot $n$ can be expressed as \cite{b7}:
\begin{equation}
	R_{\mathrm{min},n}(\mathbf{u})= \;B_r \left( \frac{2\;z(n)\;\Omega}{c}+\tau_p \right) \mathrm{PRF}, \label{equation:Rmin}
\end{equation}
where $B_r$ is the radar bandwidth, $\tau_p$ is the radar pulse duration, $c$ is the speed of light, $\rm PRF$ is the radar pulse repetition frequency, and $\Omega$ is given by: 
\begin{equation}
	\Omega=\frac{\cos(\theta_1)-\cos(\theta_2)}{\cos(\theta_1)\cos(\theta_2)}.
\end{equation}
Furthermore, the achieved radar \ac{snr} in time slot $n$ is given as follows \cite{b8}: 
\begin{equation}
	\mathrm{SNR}_n(\mathbf{P}_{\mathrm{sar}},\mathbf{u})=\frac{P_{\mathrm{sar}}(n)\; G_t\; G_r \;\lambda^3 \;\sigma_0 \;c \;\tau_p \mathrm{PRF}\;\sin^2(\theta_d)}{(4\pi)^4 \;z^3(n) \;{ k_B} \;T_o \;NF \;\;B_r \;L_{\mathrm{tot}} \;v }, \label{equation:SNR}
\end{equation}
where $\mathbf{P}_{\mathrm{sar}}=[P_{\mathrm{sar}}(1),...,P_{\mathrm{sar}}(NM)]^T$, $P_{\mathrm{sar}}(n)$ is the radar transmit power in time slot $n$,  $G_t$ and $G_r$ are the radar antenna gains for transmission and reception, respectively,  $\lambda$ is the radar wavelength, $\sigma_o$ is the backscattering coefficient, ${ k_B}$ is Boltzmann's constant, $T_o$ is the equivalent noise temperature, $NF$ is the system noise figure, and $L_{\mathrm{tot}}$
represents the total radar losses. 
\subsection{\ac{uav} Data Link}
We denote the location of the   \ac{gs} by $\mathbf{g}=[g_x,g_y,g_z]^T$. In time slot $n$, the distance between the drone and the \ac{gs} is given by:
\begin{align}
	d_n(\mathbf{u})=\sqrt{(x(n)-g_x)^2+(y(n)-g_y)^2+(z(n)-g_z)^2}.
\end{align}
We assume a \ac{los} communication link to the \ac{gs}. Thus, based on the free space path loss model\footnote{ The offline optimization framework proposed in this paper can be extended to other communication channels with small-scale and large-scale fading components by an adequate design of the link margin that accounts for the channel power fluctuations \cite{book}.}, the instantaneous throughput of the backhaul link between the \ac{uav} and the \ac{gs} is given by: 
\begin{equation}\label{eq:data_rate}
	R_n(\mathbf{P}_{\mathrm{com}},\mathbf{u})= B_{c} \; \log_2\left(1+\frac{P_{\mathrm{com}}(n) \;\gamma}{		d^2_n(\mathbf{u})}\right),
\end{equation}
where  $\mathbf{P}_{\mathrm{com}} =[P_{\mathrm{com}}(1),...,P_{\mathrm{com}}(NM)]^T$, $P_{\mathrm{com}}(n)$ is the power allocated for  communication in time slot $n$,  $B_{c}$ is the communication bandwidth, and $\gamma$ is the channel power gain at a reference distance of 1 m divided by the noise variance. To guarantee successful real-time backhauling of the sensing data collected by the drone to the \ac{gs}, the following inequality must be satisfied at all times: 
\begin{gather}
	R_n(\mathbf{P}_{\mathrm{com}},\mathbf{u}) \geq R_{\mathrm{min},n}(\mathbf{u} )+R_{\mathrm{sl}}, \forall n \in \mathbb{N}_{NM}, \label{equation:DataRate}
\end{gather}
where $R_{\mathrm{sl}}$ is the fixed data rate needed to transmit the synchronization and localization data  necessary for successful radar image processing. { In this work, we assume that the communication system operates in a licensed frequency band that is different from the one used for sensing, thus, avoiding interference that could lead to transmission failures. }Lastly, the total energy consumed by the \ac{uav} in time slot $n$ is given  by: 
\begin{equation} 
	E_{n}(\mathbf{P}_{\mathrm{com}},\mathbf{P}_{\mathrm{sar}})= \delta_t \left( \mathbf{P}_{\mathrm{com}}(n)+\mathbf{P}_{\mathrm{sar}}(n)+P_{\mathrm{prop}}\right),
\end{equation}
 {where $P_{\mathrm{prop}}$ is the propulsion power required for maintaining the operation of the drone and is given by \cite{response1}:
\ifonecolumn
 \begin{equation}\label{eq:Pprop}
P_{\mathrm{prop}}=  P_0 \left(1+\frac{3v^2}{U^2_{\mathrm{tip}}}\right)+P_I\left( \sqrt{1+\frac{v^4}{4v_0^4}}-\frac{v^2}{2v_0^2}\right)^{\frac{1}{2}} + \frac{1}{2}d_0 \rho s A v^3.
 	\end{equation}
 	\fi
Here, $P_0$ and  $P_I$ are two constants defined in \cite{response1}, $U_{\mathrm{tip}}$ is the tip speed of the rotor blade,  $d_0$ is the fuselage drag ratio, $\rho$ is the air density, $s$ is the rotor solidity, $A$ is the rotor disc area, $v_0=\sqrt{\frac{ W_u}{2 \rho A}}$ is the mean rotor induced velocity in hover, and $ W_u$ is the \ac{uav} weight in Newton. Note that, in this work, the propulsion power is constant since the drone's velocity, $v$, is fixed to allow for efficient signal processing of the \ac{sar} images.}
\section{Robust  \ac{uav} Trajectory Design Under Random Position Deviations} \label{Section:FlightDeviation}
In practice, it is difficult for lightweight drones to strictly follow a designed trajectory due to perturbations caused by weather conditions, mechanical vibrations, and localization errors \cite{jitter}. 
These errors may result in areas not covered by \ac{sar}. In this section, we first develop a model for \ac{uav} trajectory deviations and quantify their impact on the footprint of the sensing beam. Then, we propose a statistics-based robust trajectory design that suitably widens and adjusts the sensing beam footprint. This creates overlapping scans and allows covering the \ac{aoi} without gaps. 
\subsection{Trajectory Deviation Model} \label{Sec:FlightDeviationModel}
\begin{figure*}[] 
	\centering
	\begin{tabular}{cccc}
		\includegraphics[width=0.3\textwidth]{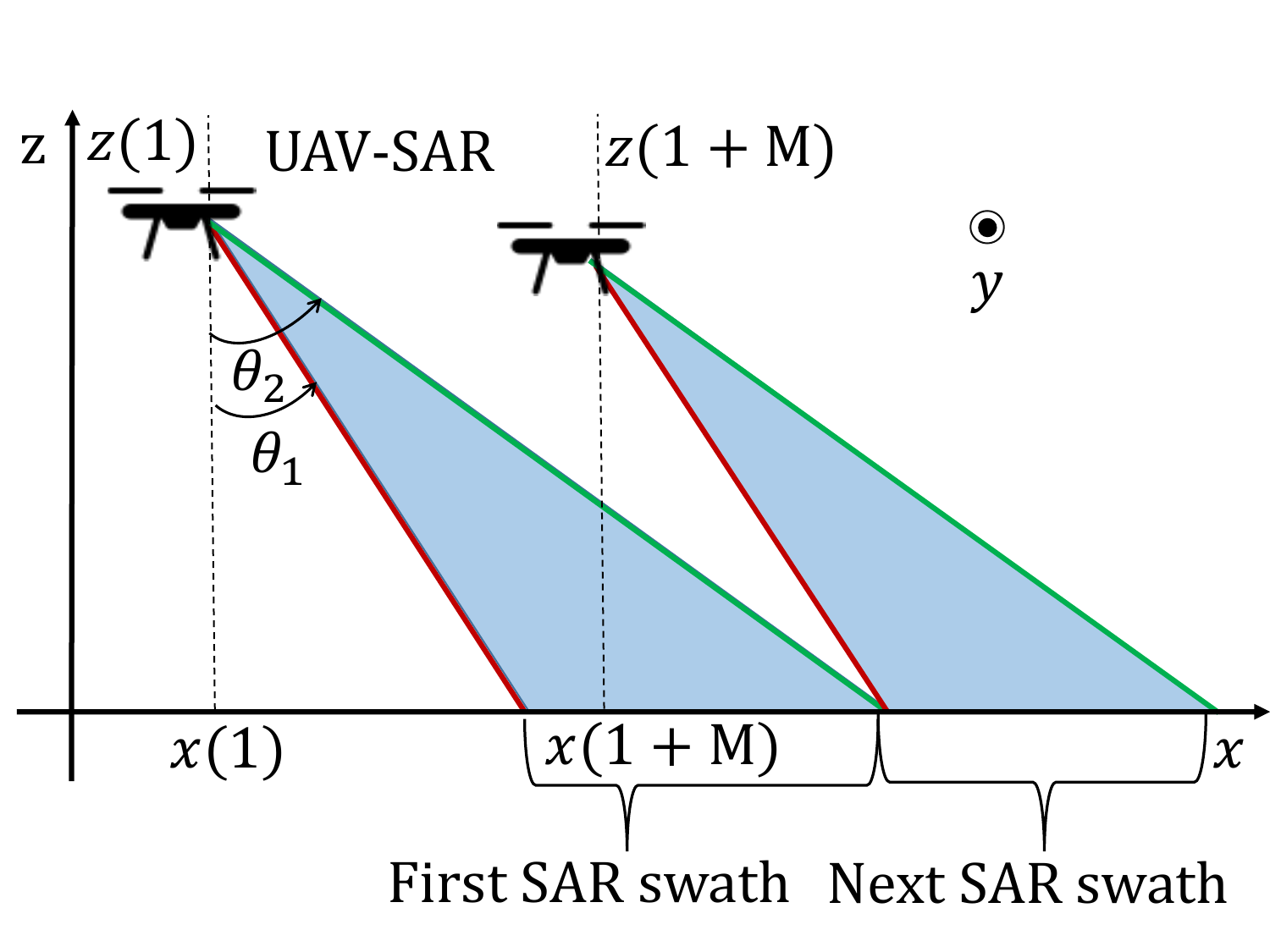} &
		\includegraphics[width=0.3\textwidth]{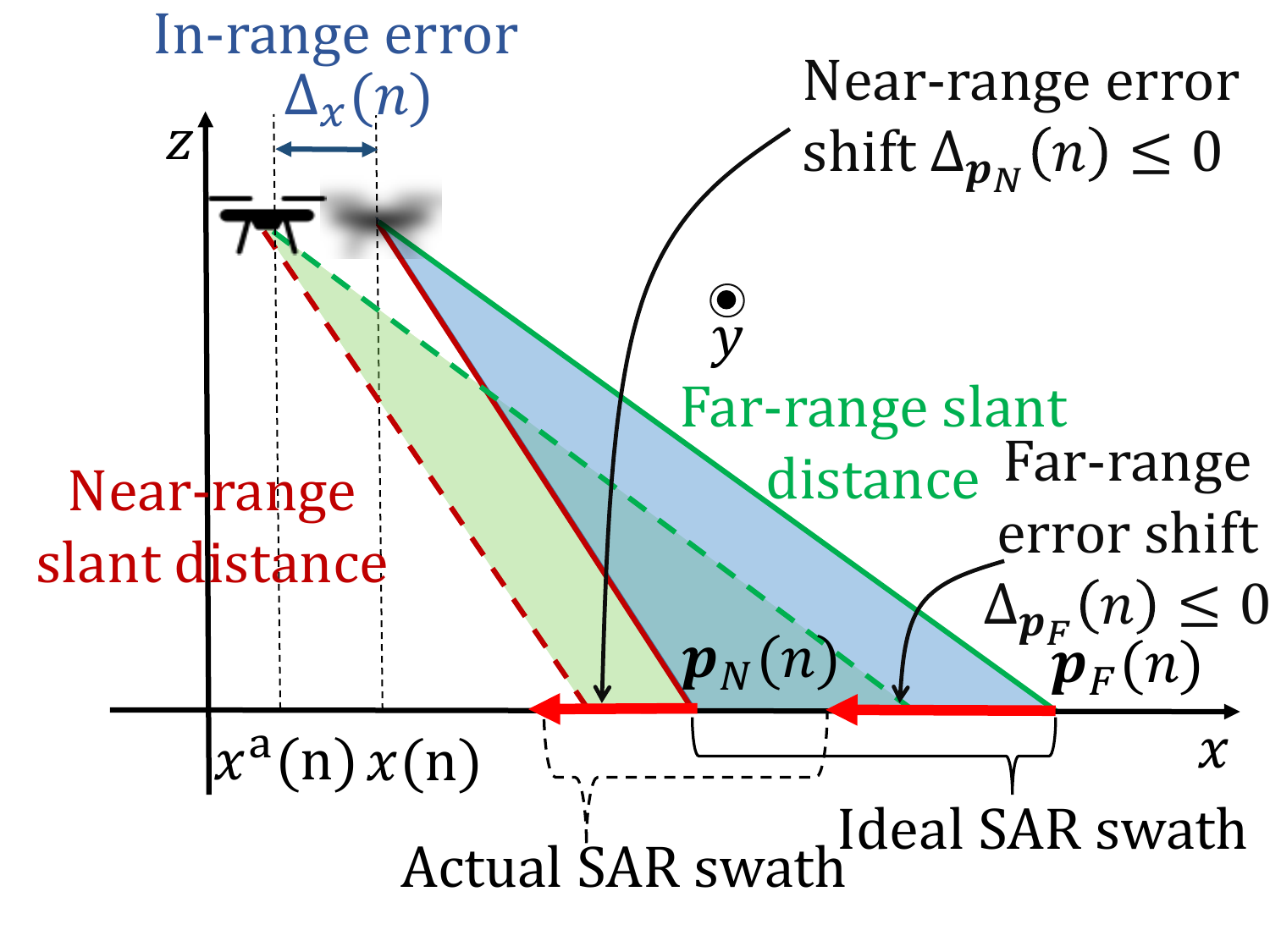} &
		\includegraphics[width=0.3\textwidth]{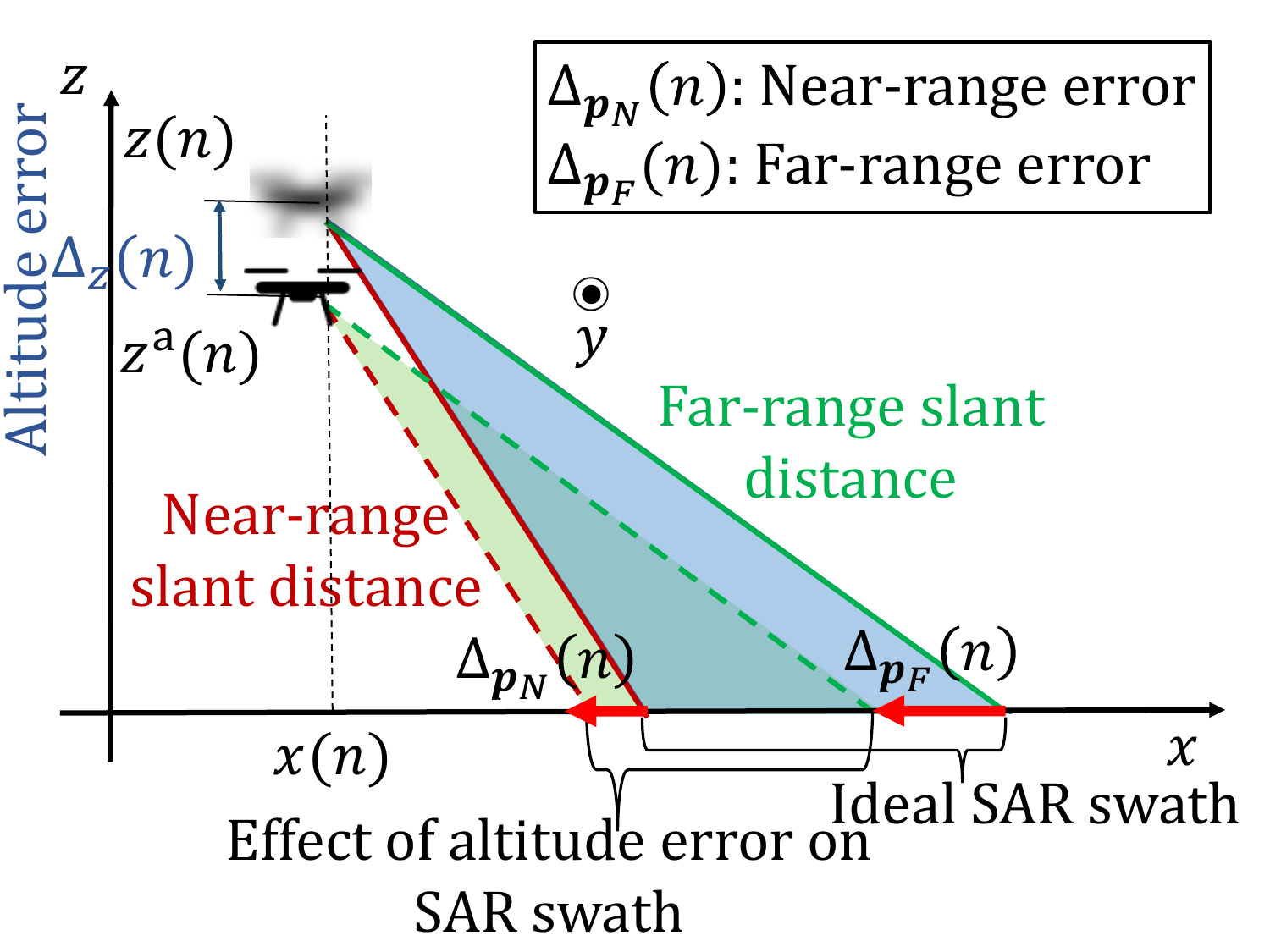} \\
		\textbf{(a)}  & \textbf{(b)} & \textbf{(c)}  \\[6pt]
	\end{tabular}
	\begin{tabular}{cccc}
		\includegraphics[width=0.3\textwidth]{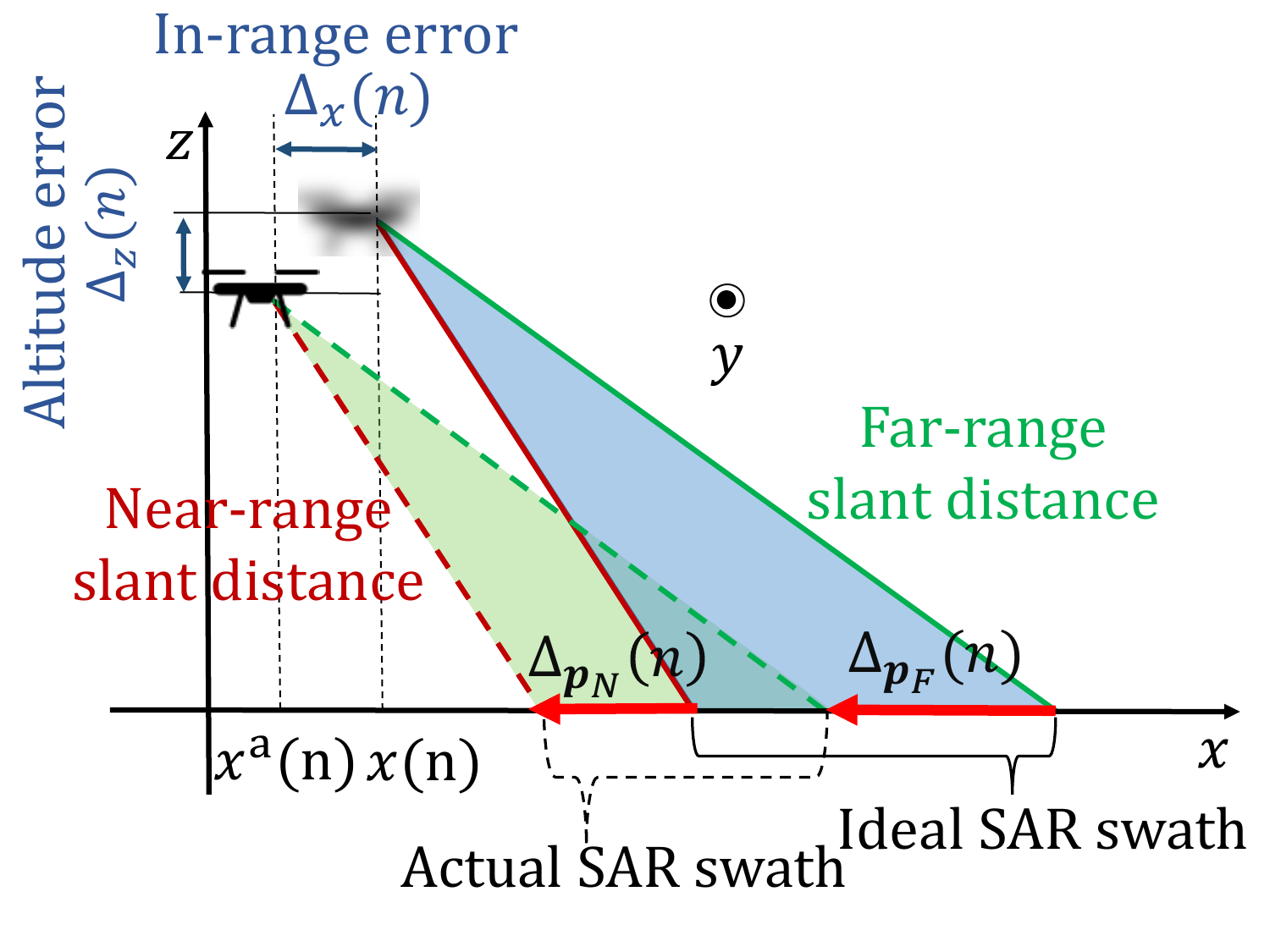} &
		\includegraphics[width=0.3\textwidth]{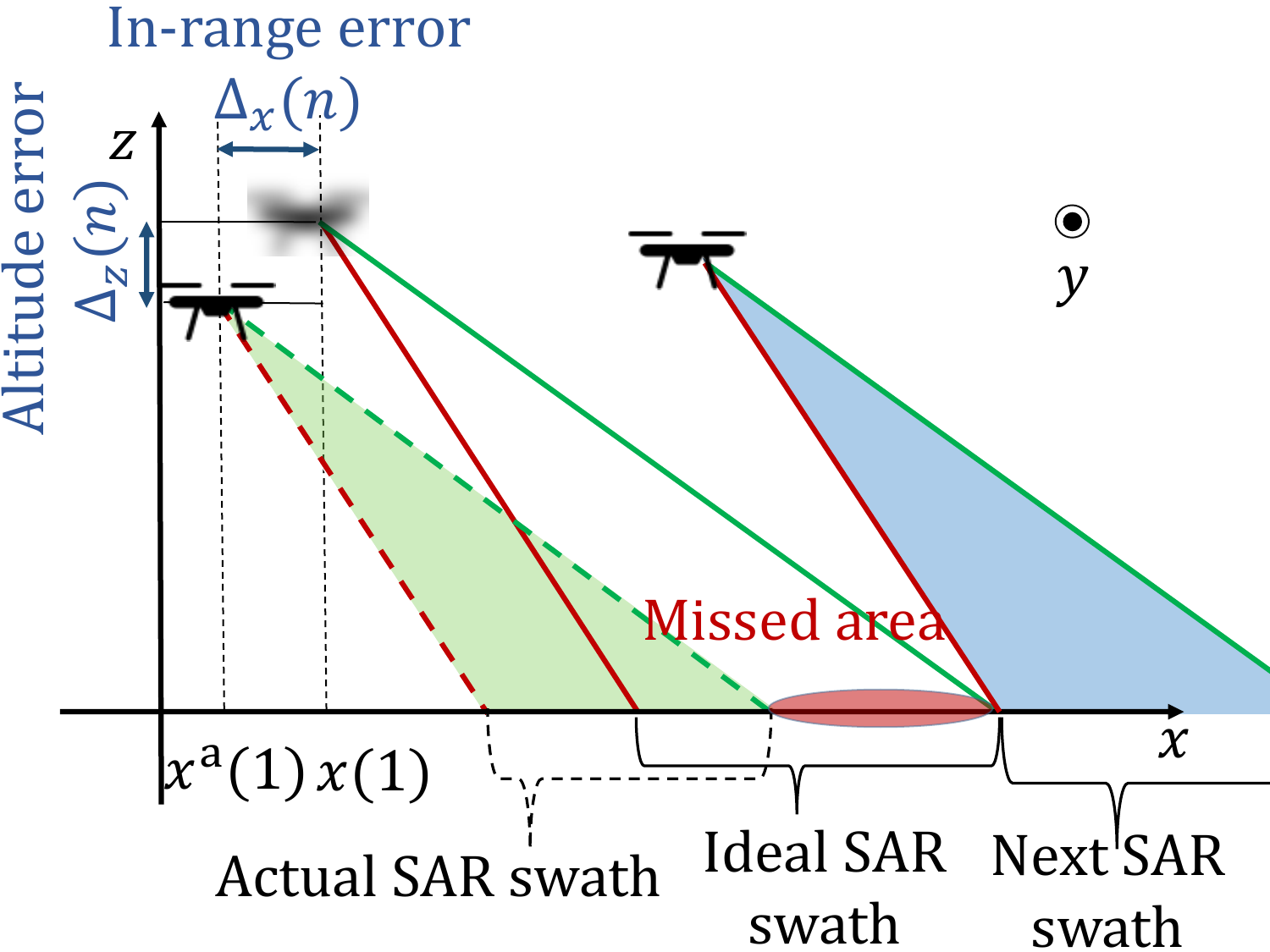} &
		\includegraphics[width=0.3\textwidth]{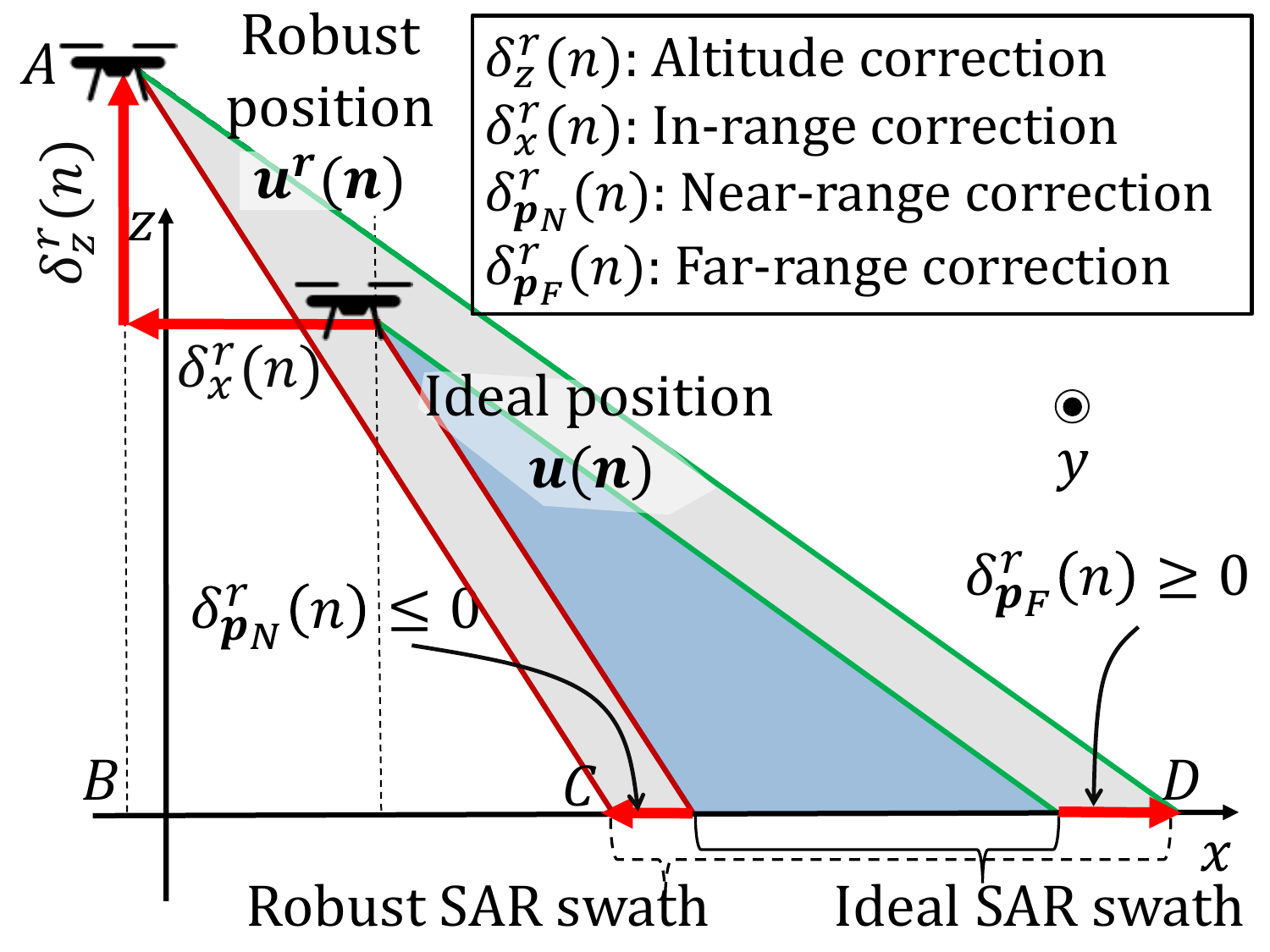}\\
		\textbf{(d)}  & \textbf{(e)}    & \textbf{(f)} \\[6pt]
	\end{tabular}
	\caption{ Robust \ac{uav} trajectory design. \textbf{(a)} Successive azimuth scans for an ideal trajectory.
		\textbf{(b)} {Effect of in-range error on the ground \ac{sar} swath.
			\textbf{(c)} Effect of altitude error on the ground \ac{sar} swath.
			\textbf{(d)} Combined effect of in-range and altitude errors on the ground \ac{sar} swath.}
		\textbf{(e)} Uncovered areas due to \ac{uav} flight deviations. \textbf{(f)} Proposed  robust \ac{uav}-\ac{sar} trajectory.  }
	\label{fig:Overlapp}
\end{figure*}
If the actual trajectory is different from the ideal linear trajectory $\mathbf{u}$, described by $\rm C1-C4$, the actual footprint of the sensing beam is also affected in terms of both its on-ground position and width. Let position vectors  $(\mathbf{x}^a,\mathbf{y}^a,\mathbf{z}^a) \in \mathbb{R}^{NM\times 1}$ describe the actual trajectory, denoted by $\mathbf{u}^a=[\mathbf{u}^a(1),...,\mathbf{u}^a(NM)]$, where  the drone's actual position in time slot $n$ is given by $\mathbf{u}^a(n)=[x^a(n),y^a(n),z^a(n)]^T$.  In what follows, we quantify the change in the actual footprint of the sensing beam due to random \ac{uav} trajectory deviations.\par
\subsubsection{Cross-range error}\sloppy The deviation in $y$-direction, denoted by $\boldsymbol{\Delta}_y = \mathbf{y}^a-\mathbf{y}=$ $[\Delta_y(1), ..., \Delta_y(NM)]^T \in \mathbb{R}^{NM\times 1}$, has the least effect on the \ac{sar} ground coverage. Although this error can affect the \ac{sar} resolution by influencing the illumination time of a given reference point on the ground, it has no effect on the radar coverage thanks to the radar beam width in $y$-direction. In fact, the main idea behind \ac{sar} systems is to use this direction to create the synthetic aperture such that any reference point on the ground is illuminated multiple times from different angles along the $y$-axis \cite{tutorial}. Therefore, in what follows, we neglect the effect of this error on the radar coverage and  focus on \ac{uav} trajectory deviations in the $xz$-plane.\par
\subsubsection{In-range error} This error is caused by the \ac{uav} position deviation in $x$-direction and is denoted by vector $\boldsymbol{\Delta}_x = \mathbf{x}^a-\mathbf{x} =$ $ [\Delta_x(1), ..., \Delta_x(NM)]^T$ $\in$ $ \mathbb{R}^{NM \times 1 }$, see Figure \ref{fig:Overlapp}b. As shown in the figure, this error causes near-range and far-range shifts of the ground footprint of the sensing beam along the $x$-axis. To this end,  in the $xz-$plane and in time slot $n$, we define point $\mathbf{p}_{N}(n)=[x(n)+c_1z(n),y(n),0]^T$ representing the intersection between the near-range slant distance line and the $x$-axis and point $\mathbf{p}_F(n)=[x(n)+c_2z(n),y(n),0]^T$ representing the intersection between the far-range slant distance line and the $x$-axis. Both points are defined \ac{wrt}  the drone's ideal trajectory, see Figure \ref{fig:Overlapp}b. Any deviation from the ideal \ac{uav} position in the $xz$-plane  results in a near-range error (shift of point $\mathbf{p}_N(n)$ along the $x$-axis), denoted by $\boldsymbol{\Delta}_{\mathbf{p}_N}=[\Delta_{\mathbf{p}_N}(1), ...,\Delta_{\mathbf{p}_N}(NM)]^T \in \mathbb{R}^{NM\times1}$, and a far-range error (shift of point $\mathbf{p}_F(n)$ along the $x$-axis), denoted by  $\boldsymbol{\Delta}_{\mathbf{p}_F}=[\Delta_{\mathbf{p}_F}(1), ...,\Delta_{\mathbf{p}_F}(NM)]^T \in \mathbb{R}^{NM\times1}$, of the sensing beam footprint. A deviation in $x$-direction,   causes identical near-range and far-range errors such that, in time slot $n$, $\Delta_{\mathbf{p}_F}(n)=\Delta_{\mathbf{p}_N}(n)=\Delta_x(n)$. Therefore, an in-range error does not affect the width of the \ac{sar} swath, but results in a translation of the footprint of the sensing beam along the $x$-axis as shown in  Figure \ref{fig:Overlapp}b. \par
\subsubsection{Altitude error} A deviation of the trajectory in $z$-direction, given by $\boldsymbol{\Delta}_z = \mathbf{z}^a-\mathbf{z} =  [\Delta_z(1), ..., \Delta_z(NM)]^T \in \mathbb{R}^{NM \times 1}$, results in non-equal far-range and near-range errors. In fact, an altitude error in time slot $n$, denoted by $\Delta_z(n)$, results in $\Delta_{\mathbf{p}_N}(n)=c_1 \Delta_z(n) $ and  $\Delta_{\mathbf{p}_F}(n)=c_2 \Delta_z(n)$. In contrast to in-range errors,  altitude errors induce a change in both the width and the position of the sensing beam footprint, see Figure \ref{fig:Overlapp}c. \par

As shown in Figure \ref{fig:Overlapp}d, in time slot $n$, the combined in-range and altitude errors result in a near-range error of the footprint of the sensing beam given by: 
\begin{equation}
	\Delta_{\mathbf{p}_N}(n)=\Delta_x(n)+c_1 \Delta_z(n). \label{eq:Delta_N}
\end{equation}
Similarly,  in time slot $n$, the total far-range error is given by: 
\begin{equation}
	\Delta_{\mathbf{p}_F}(n)=\Delta_x(n)+c_2 \Delta_z(n).\label{eq:Delta_F}
\end{equation}
\subsubsection{Trajectory deviation model}
In this work,  based on measurements in \cite{deviation,motioncompensation}, we propose to model the \ac{uav}-\ac{sar} trajectory deviations as independent Gaussian random variables such that, in each time slot $n$, $ \Delta_x(n)\sim\mathcal{N}( o_x,\sigma)$ and $\Delta_z(n)\sim\mathcal{N}( o_z,\sigma)$, $\forall n$. The standard deviation, denoted by $\sigma$, depends on the stability of the \ac{uav} in the presence of external perturbations caused, e.g., by adverse weather conditions. The offsets in $x$- and $z$-direction, denoted by $o_x$ and $o_z$, respectively, represent a possible constant error between the \ac{uav}'s actual trajectory and its ideal intended trajectory and depend on the \ac{uav} localization accuracy. For the special case, $o_x=o_z=0$, the \ac{uav} trajectory deviations reduce to random fluctuations around the ideal trajectory, also referred to as \ac{uav} jittering \cite{jitter2}.

\subsection{Proposed Robust Trajectory Design}
Both positive near-range errors, i.e., $\Delta_{\mathbf{p}_N}(n)> 0$, and negative far-range errors, i.e.,  $\Delta_{\mathbf{p}_F}(n)< 0$, result in uncovered areas, whereas  the reverse result in an overlap 
between successive scans, see Figure \ref{fig:Overlapp}e. For robust trajectory design, uncovered areas are avoided by tolerating an overlap between adjacent azimuth scans. This can be achieved by shifting the actual footprint of the sensing beam such that the near-range and far-range errors are compensated. Since excessive overlap between scans reduces the \ac{sar} coverage, the overlap should be limited to the degree needed for compensation of \ac{uav} trajectory deviations. To this end, we define the coverage reliability  level $r$, $0\leq r \leq 1$, as the probability of having no gaps between successive \ac{sar}  {  azimuth scans}. The compensation {that is} required in the near range of the sensing beam footprint { to achieve} coverage reliability level $r$ is denoted by $ \boldsymbol{\delta}_{\mathbf{p}_N}^r=[\delta_{\mathbf{p}_{\scriptscriptstyle N}}^r(1), ..., \delta_{\mathbf{p}_N}^{r}(NM)]^T$ $\in \mathbb{R}^{NM\times1}$, and is obtained as:
\begin{equation}
	\delta_{\mathbf{p}_N}^{r}(n)\stackrel{}{=} \argmin_{s}  \Big\{ |s| \Big| s \in \mathbb{R}, \mathbb{P}\left(\Delta_{\mathbf{p}_N}(n)+ s\leq 0 \right) \geq r \Big\}. \label{eq:Near-rangeDefinition}
\end{equation}
Similarly, the far-range compensation, which we denote by $ \boldsymbol{\delta}_{\mathbf{p}_F}^{r}=$ $[\delta_{\mathbf{p}_F}^{r}(1), ..., \delta_{\mathbf{p}_F}^{r}(NM)]^T$ $\in \mathbb{R}^{NM\times1}$, is obtained as: 
\begin{equation}
	\delta_{\mathbf{p}_F}^{r}(n)\stackrel{}{=} \argmin_{s} \Big\{ |s|\Big| s \in \mathbb{R}, \mathbb{P}\left(\Delta_{\mathbf{p}_F}(n)+ s\geq 0 \right) \geq r \Big\}. \label{eq:Far-rangeDefinition}
\end{equation}
{Note that, on the one hand, to avoid coverage holes, coverage reliability values close to 1 are desirable. On the other hand, if $r$ is chosen too close to 1, this might result in a reduction of the total coverage, as a larger swath overlap is required to compensate for random trajectory deviations. Based on our simulations, $r=0.95$ provides a good compromise between coverage reliability and total coverage.}\newline Next, given the trajectory deviation statistics, we derive analytical expressions for the near-range and far-range compensations. 
\begin{proposition} \label{proposition:Near-Far-Range}
	The near-range and far-range compensations given in (\ref{eq:Near-rangeDefinition}) and (\ref{eq:Far-rangeDefinition}), respectively, can be obtained in time slot $n$ as follows: 
	\begin{align}
		\delta_{\mathbf{p}_N}^{r}(n) &=-\Big[\mathrm{erf}^{-1}(2r-1)\sigma \sqrt{2(1+{c_1}^2)}+o_x+c_1 o_z\Big]^+, \label{eq:deltaN}\\ 
		\delta_{\mathbf{p}_F}^{r}(n) &=\Big[\mathrm{erf} ^{-1}(2r-1)\sigma \sqrt{2(1+{c_2}^2)}-o_x-c_2 o_z\Big]^+.\label{eq:deltaF}
	\end{align}
\end{proposition}
\begin{proof}
 Please refer to Appendix \ref{app:P.1}.
\end{proof}
\begin{remark}
	Note that $\delta_{\mathbf{p}_N}^{r}(n) \leq 0 $,  whereas $\delta_{\mathbf{p}_F}^{r}(n)\geq 0 $, $\forall n \in \mathbb{N}_{NM}$, which means that the actual footprint of the radar beam is widened, see Figure \ref{fig:Overlapp}f. This creates overlapping scans that prevent gaps between adjacent scans. 
\end{remark}
The proposed near-range and far-range compensations can be used to adjust the \ac{uav} position. In the following proposition, we translate the near-range and  far-range compensation of the beam footprint  to an  in-range ($x$-position) adjustment, denoted by $\boldsymbol{\delta}_x^{r}=[\delta_x^{r}(1), ...,\delta_x^{r}(NM)]^T \in \mathbb{R}^{NM\times1}$,  and  an altitude ($z$-position) adjustment, denoted by $\boldsymbol{\delta}_z^{r}=[\delta_z^{r}(1), ...,\delta_z^{r}(NM)]^T\in \mathbb{R}^{NM\times1}$, as shown in Figure \ref{fig:Overlapp}f.
\begin{proposition}
	\label{theo}Given the near-range and far-range compensation of the beam footprint, $\boldsymbol{\delta}_{\mathbf{p}_N}^{r}$ and $\boldsymbol{\delta}_{\mathbf{p}_F}^{r}$, the corresponding in-range and altitude adjustment of the \ac{uav} position are respectively given by: \begin{gather}
		\boldsymbol{\delta}_x^{r}=\boldsymbol{\delta}_{\mathbf{p}_N}^{r}-c_1\; \boldsymbol{\delta}_z^{r}, \label{theo:along-range}\\          \boldsymbol{\delta}_z^{r}= \frac{\boldsymbol{\delta}_{\mathbf{p}_F}^{r}-\boldsymbol{\delta}_{\mathbf{p}_N}^{r}}{c_2-c_1}.\label{eq:shift_z}
	\end{gather}
\end{proposition}

\begin{proof}{ Based on the geometry of the problem, In time slot $n$,} given $\delta_{\mathbf{p}_N}^{r}(n)\leq 0$ and $\delta_{\mathbf{p}_F}^{r}(n)\geq 0$, we shift the \ac{uav} position along the $x$-axis by $\delta_x^{r}(n)\leq 0$, and along the  $z$-axis by $\delta_z^{r}(n)\geq 0$, see Figure \ref{fig:Overlapp}f. The required \ac{uav} position adjustments can be derived {based on the geometry of the problem} by solving the following system of equations:
	\ifonecolumn
	\begin{align}
		\begin{cases}\label{eq:Syseq}
			c_1=\frac{c_1z(n)+\delta_{\mathbf{p}_N}^{r}(n)-\delta_x^{r}(n) }{z(n)+\delta_z^{r}(n)},\\
			c_2=\frac{c_2z(n)+\delta_{\mathbf{p}_F}^{r}(n)-\delta_x^{r}(n)}{z(n)+\delta_z^{r}(n)},
			\end{cases}
	\end{align}
	\fi
	where the first and second lines of (\ref{eq:Syseq}) are obtained from triangles ($ABC$) and ($ABD$) in Figure \ref{fig:Overlapp}f, respectively.  
\end{proof}
For robust \ac{uav}-\ac{sar} trajectory design, the \ac{uav} position is adjusted in each time slot to achieve a coverage reliability level $r$ in the presence of \ac{uav} trajectory deviations. Consequently, the robust trajectory is denoted by  $\textbf{u}^r=[\textbf{u}^r(1),...,\textbf{u}^r(NM)]$, where  $ \mathbf{u}^{r}(n)$ is the adjusted position in time slot $n$ and is given by:
\ifonecolumn 
\begin{align}
	\mathbf{u}^{r}(n)&=\left[ x^{r}(n),y(n),z^{r}(n)\right]^T=\left[ x(n)+\delta_x^{r}(n),y(n),z(n)+\delta_z^{r}(n)\right]^T.\label{eq:UAVcorrectedpos}
\end{align}
\fi  

\section{Robust Resource Allocation and Trajectory Optimization Framework} \label{Sec:ResourceAllocation}
In this section, we formulate an optimization problem for joint \ac{3d} trajectory and resource allocation design. {The formulated problem is solved offline where random \ac{uav} trajectory deviations are accounted for via the  robust trajectory given in (\ref{eq:UAVcorrectedpos}). In particular}, for a given coverage reliability level $r$, we provide a low-complexity sub-optimal solution  of the formulated problem based on SCA. Moreover, to assess the performance of the proposed sub-optimal scheme, we derive an upper bound for the optimal solution based on monotonic optimization theory {\cite{1}}. 
\subsection{Problem Formulation}
Our objective is the maximization of the drone ground coverage $C$ while satisfying the constraints imposed by the radar and communication subsystems. We assume Gaussian trajectory deviation as  described in Section \ref{Sec:FlightDeviationModel}. To achieve a coverage reliability level $r$, the \ac{uav} {follows} the robust trajectory $\mathbf{u}^r$ provided in (\ref{eq:UAVcorrectedpos}). We optimize the drone trajectory $
\{ \mathbf{z},\mathbf{x},N\}$ and the power allocation $\{\mathbf{q}, \mathbf{P_{\mathrm{com}}}, \mathbf{P_{\mathrm{sar}}}\}$, where vectors $\mathbf{z}$, $\mathbf{x}$, $\mathbf{q}$, $ \mathbf{P_{\mathrm{com}}}$, and $\mathbf{P_{\mathrm{sar}}}$,  represent the collections of all $z(n), x(n), q(n), P_{\rm com}(n),$ and $ P_{\rm sar}(n), \forall n \in \mathbb{N}_{NM}$, respectively, and $q(n)$ is the energy { remaining} in the drone's battery in time slot $n$. To this end, the following optimization problem is formulated:
\begin{alignat*}{2} 
	&\!(\mathrm{P.1}):\max_{\mathbf{z},\mathbf{x},\mathbf{q}, \mathbf{P_{\mathrm{sar}}}, \mathbf{P_{\mathrm{com}}},N}  C(\mathbf{u}^{r})     & \qquad&  \\
	&\text{s.t.} \hspace{2mm} \mathrm{C1-C4,} &   & \\ 
	&  \mathrm{C5: } \;P_{\mathrm{sar}}(n+1)=P_{\mathrm{sar}}(n), \forall n \in \mathcal{A},                &      &  \\ & \mathrm{C6}: \;z_{\mathrm{min}}\leq z^{r}(n) \leq z_{\mathrm{max}}, \forall n \in \mathbb{N}_{NM},                 &      &  \\&  \mathrm{C7}:\mathrm{SNR}_n(\mathbf{P}_{\mathrm{sar}}, \mathbf{u}^{r}) \geq \mathrm{SNR_{min}},  \forall n \in \mathbb{N}_{NM} ,               &      &     
	\\
	&  \mathrm{C8}:  R_n (\mathbf{P}_{\mathrm{com}}, \mathbf{u}^{r}) \geq R_{\mathrm{min},n}(\mathbf{u}^{r})+R_{\mathrm{sl}},  \forall n \in \mathbb{N}_{NM} ,            &      &     
	\\
	&    \mathrm{C9}:0\leq P_{\mathrm{sar}}(n)\leq P_{\mathrm{sar}}^{\mathrm{max}}, 0\leq P_{\mathrm{com}}(n)\leq P_{\mathrm{com}}^{\mathrm{max}},  \forall n \in \mathbb{N}_{NM},             &      & \\
	& \mathrm{C10}: q(1)= q_{\mathrm{start}},\;q(n) \geq 0,  \forall n \in \mathbb{N}_{NM} ,     & & \\
	& \mathrm{C11}: q(n+1) = q(n) - E_{n}(\mathbf{P}_{\mathrm{sar}}, \mathbf{P}_{\mathrm{com}}) , \forall n \in { \mathbb{N}_{NM-1}},      & &  \\
	& \mathrm{C12}: N \in \mathbb{N}.     & &  
\end{alignat*}
Constraints $\mathrm{C1-C4}$ define the shape of the \ac{uav} ideal trajectory as described in (\ref{eq:turn1})-(\ref{eq:turn4}). Constraint $\mathrm{C5}$ ensures constant radar power along each azimuth scan.  Constraint $\mathrm{C6}$ determines the allowed range for the \ac{uav}'s operational altitude when the drone follows the proposed robust trajectory. Constraint $\mathrm{C7}$ ensures the minimum \ac{snr},  $\mathrm{SNR_{min}}$, required for \ac{sar} imaging when the drone follows the robust trajectory $ \mathbf{u}^{r}$ given by (\ref{eq:UAVcorrectedpos}). Constraint $\mathrm{C8}$ guarantees successful real-time transmission of the \ac{sar} data collected by the  \ac{uav} to the \ac{gs}. Constraint $\mathrm{C9}$ ensures that the radar and communication transmit powers are non-negative and do not exceed the maximum permissible levels. {Constraint $\mathrm{C10}$ specifies the energy available in the \ac{uav}  battery,  $q_{\mathrm{start}}$, at the start of the mission, and ensures that the energy remaining in the battery during operation,  $q(n)$, does not become negative. In addition, constraint $\mathrm{C11}$   updates the \ac{uav} battery level in time slot $n+1$ by subtracting the total energy consumed during time slot $n$, denoted by $E_n(\mathbf{P}_{\mathrm{sar}},\mathbf{P}_{\mathrm{com}})$, from $q(n)$.}
Constraint $\mathrm{C12}$ defines the number of azimuth scans realized by the \ac{uav} which is an integer.\par 
\par
Problem $\rm (P.1)$ is a non-convex \ac{minlp}  due to non-convex constraint $\mathrm{C8}$ and integer variable constraint $\mathrm{C12}$. The latter does not only affect the objective function but also the dimension of all  other optimization variables. Therefore, $N$ cannot be optimized jointly with the remaining variables. Furthermore, non-convex \acp{minlp} are usually very difficult to solve. Nevertheless, in the following, we provide a low-complexity sub-optimal solution to problem $\rm (P.1)$ based on SCA. 
\subsection{Low-complexity Sub-optimal Solution of Problem $\mathrm{(P.1)}$} \label{Sec:Sub-optimal}
In this subsection, we provide a low-complexity sub-optimal solution to problem $\rm (P.1)$ based on \ac{sca}. First, we fix the number of azimuth scans $N$ and denote the resulting problem by $\mathrm{(P.2)}$:
\begin{alignat*}{2}
	&\!\mathrm{(P.2)}:\max_{\mathbf{z},\mathbf{x},\mathbf{q}, \mathbf{P_{\mathrm{sar}}}, \mathbf{P_{\mathrm{com}}}} C(\mathbf{u}^{r})      & \qquad&  \\
	&\text{s.t.} \hspace{3mm} \mathrm{C1-C11.} &   & \; 
\end{alignat*}
However, even for a given $N$, problem $\mathrm{(P.2)}$ is still non-convex and difficult to solve. Thus, we approximate problem $\rm (P.2)$ as a convex problem and use iterative \ac{sca}  to solve it. Then, we perform a search for the optimal number of azimuth scans $N^*$.  As a first step, we  transform constraint $\rm C7$, which involves a cubic function of drone altitude $z$, using second-order cone programming. To this end, we introduce  slack variables $\psi(n)$ stacked in vector $\boldsymbol{\psi}=[\psi(1),...,\psi(NM)]^T\in \mathbb{R}^{NM\times 1}$,  and reformulate $\rm C7$ as follows: 
\begin{align}\label{eq:beta}
\mathrm{C7a}&:\begin{bmatrix} 
		\psi(n) & z^{r}(n)\\
		z^{r}(n) & 1 
	\end{bmatrix} \succeq \mathbf{0},\forall n \in \mathbb{N}_{NM},\\
	 \mathrm{C7b}&: \begin{bmatrix}
		P_{\rm sar}(n) \beta & \psi(n) \\
		\psi(n) & z^{r}(n)
	\end{bmatrix} \succeq \mathbf{0},   \forall n \in \mathbb{N}_{NM} , \\ 
	\mathrm{C7c}&:\psi(n) \geq 0, \forall n \in \mathbb{N}_{NM},
\end{align}
where $\beta=\frac{ G_t\; G_r\;\lambda^3 \;\sigma_0 \;c \;\tau_p \mathrm{PRF}\;\sin^2(\theta_d)}{(4\pi)^4 \;{ k_B} \;T_o \;NF \;\;B_r \;L_{\mathrm{tot}} \;v\; \mathrm{SNR_{min}} }$. Furthermore, non-convex constraint $\rm C8$ can be  rewritten as follows:
\begin{equation} 
	{\rm C8}: \left( A2^{\alpha z^{r}(n)} -1\right)	d^2_n(\mathbf{u}^{r})\leq P_{\mathrm{com}}(n)\gamma, \forall n \in \mathbb{N}_{NM},\label{eq:constraintC8}
\end{equation}
where $A=2^{\frac{B_r\;\tau_p\;\mathrm{PRF} }{B_c}}$ and $\alpha=\frac{2 \;\Omega \;B_r\; \mathrm{PRF}}{c \;B_c}$. In time slot $n$, let function $H_n(\mathbf{u}^{r})$ be equal to the left-hand side of (\ref{eq:constraintC8}). Function $H_n$ is non-convex \ac{wrt} to optimization variables $x$ and $z$. Therefore, in each \ac{sca} iteration $j\in \mathbb{N}$, we use a Taylor series approximation around points $\{z^{(j)}(n), x^{(j)}(n) \}\in \mathbb{R}$ to convexify constraint $\rm C8$. This is achieved by replacing function $H_n$ with a global overestimate, denoted by $\overline{H}_n$. To this end, in time slot $n$, we first expand function $H_n$ as follows: 
\ifonecolumn
\begin{align}
	H_n(\mathbf{u}^{r})=h_1(z(n))h_2(z(n))+h_1(z(n))h_3(x(n))+ h_1(z(n))(y(n)-g_y)^2, \label{eq:expansion}
\end{align}
\fi
where  functions $h_1,h_2,$ and $h_3$ are convex differentiable functions that are given by: 
\begin{align}
	h_1(z(n))&=A2^{\alpha z^{r}(n)  }-1,\\
	h_2(z(n))&= (z^{r}(n)-g_z)^2,\\
	h_3(x(n))&= (x^{r}(n)-g_x)^2.
\end{align}
Yet, terms $h_1(z(n))h_2(z(n))$ and $h_1(z(n))h_3(x(n))$ contain coupling and are non-convex w.r.t. optimization variables $x$ and $z$. To overcome this obstacle, we rewrite them as follows: 
\ifonecolumn
\begin{align}
	h_1(z(n))h_2(z(n))&= \frac{1}{2} \big(   h_1(z(n))+h_2(z(n))\big)^2 - \frac{1}{2}  h_1^2(z(n))-\frac{1}{2}  h_2^2(z(n)),\label{eq:h1h2}\\ 
	h_1(z(n))h_3(x(n))&= \frac{1}{2} \big(   h_1(z(n))+h_3(x(n))\big)^2 - \frac{1}{2}  h_1^2(z(n))-\frac{1}{2}  h_3^2(x(n)). \label{eq:h3h4} 
\end{align}
\fi
Equations (\ref{eq:h1h2}) and (\ref{eq:h3h4}) are differences of convex functions, where functions $f_i=\frac{1}{2}h_i^2,  i \in \{1,2,3\}$, represent an obstacle to solving problem $\rm (P.2)$. Therefore, we replace them with their global underestimates around points $\{z^{(j)}(n), x^{(j)}(n) \}$ based on Taylor series approximation:
\ifonecolumn
\begin{align}\label{eq:approximation}
	\underline{ f_i}(z(n)) &=f_i(z^{(j)}(n)) + f_i'(z^{(j)}(n))(z(n)-z^{(j)}(n)),\forall i \in \{1,2\},\\
	\underline{ f_3}(x(n)) &=f_3(x^{(j)}(n)) + f_3'(x^{(j)}(n))(x(n)-x^{(j)}(n)), \forall n. 
\end{align}
\fi
Consequently, the global overestimate for function $\overline{H}_n$ around $\{z^{(j)}(n), x^{(j)}(n) \}$  is given by:
\ifonecolumn
\begin{align}
\overline{H}_n(\mathbf{u}^{r})=& \frac{1}{2} \big(   h_1(z(n))+h_2(z(n))\big)^2 +\frac{1}{2} \big(   h_1(z(n))+h_3(x(n))\big)^2 - \nonumber \\ &2\underline{ f_1}(z(n))  - \underline{ f_2}(z(n))- \underline{ f_3}(x(n)) +h_1(z(n))(y(n)-g_y)^2,\forall n.\label{eq:expansion2}
\end{align}
\fi
In the final step, we introduce two new slack vectors $\mathbf{t}=[t(1), ...,t(NM)]^T \in \mathbb{R}^{NM\times1}$,  $\mathbf{o}=[o(1), ...,o(NM)]^T\in \mathbb{R}^{NM\times1}$, and two new constraints  $\rm \widetilde{C8a}$ and $\rm \widetilde{C8b}$ to reformulate constraint $\rm C8$ as follows: 
\ifonecolumn
\begin{align}
	&\widetilde{ \rm C8}:  \frac{1}{2} t^2(n) +\frac{1}{2} o^2(n)- 2\underline{ f_1}(z(n))- \underline{ f_2}(z(n))- \underline{ f_3}(x(n))  +\nonumber \\ &\hspace{7mm}h_1(z(n))(y(n)-g_y)^2\leq P_{\rm com}(n)\gamma, \forall n \in \mathbb{N}_{NM},\\
	&\widetilde{ \rm C8a}:  h_1(z(n))+h_2(z(n)) \leq  t(n),\forall n \in \mathbb{N}_{NM},\\
	&\widetilde{ \rm C8b}:  h_1(z(n))+h_3(z(n)) \leq  o(n),\forall n \in \mathbb{N}_{NM}.
\end{align}
\fi
{ Let $\mathbf{z}^{(j)}\in \mathbb{R}^{NM\times1}$ and $ \mathbf{x}^{(j)}\in \mathbb{R}^{NM\times1}$ be the collections of all $z(n)$ and $x(n), \forall n$, respectively.} Problem $\mathrm{(P.2)}$ can be approximated around the solution $\{\mathbf{z}^{(j)},\mathbf{x}^{(j)}\}$ by the following convex optimization problem:
\begin{subequations}
	\begin{alignat*}{2}
		&\!\mathrm{(P.3)}:\max_{\mathbf{z},\mathbf{x},\mathbf{q}, \mathbf{P_{\mathrm{sar}}}, \mathbf{P_{\mathrm{com}}},\boldsymbol{\psi},\mathbf{t},\mathbf{o}}   C(\mathbf{u}^{r})        & \qquad&  \\
		&\text{s.t.} \hspace{3mm}\mathrm{C1-C6,C7a,C7b,C7c,\widetilde{C8},\widetilde{C8a},\widetilde{C8b}, C9-C11}. &   & \; 
	\end{alignat*}
\end{subequations}

\begin{algorithm}[]
	\caption{Successive Convex Approximation (SCA)}\label{algorithm1}
	\begin{algorithmic}[1] 
		\State Initialization: Set iteration index $j=1$,  point {$\{\mathbf{z}^{(1)},\mathbf{x}^{(1)}\}$,} coverage $C^{(1)}$ by solving $\rm (P.3)$ around point $ \{\mathbf{z}^{(1)},\mathbf{x}^{(1)}\}  $, and error tolerance $0< \epsilon \ll 1$.
		\State \textbf{repeat}
		\State$ $Set $ j=j+1$ \Comment{{ Increment iteration index}}
		\State $  $Get coverage $C^{(j)}$ and solution $ \{\mathbf{z},\mathbf{x},\mathbf{q}, \mathbf{P_{\mathrm{sar}}}, \mathbf{P_{\mathrm{com}}}\}$ by solving $\rm (P.3)$  around point $ \{\mathbf{z}^{(j-1)},\mathbf{x}^{(j-1)}\}$       
		\State$ $Set $ \mathbf{ z}^{(j)}= \mathbf{z}, \mathbf{ x}^{(j)}=\mathbf{x}${\Comment{ Update solution}}
		\State \textbf{until} $\big |\frac{C^{(j)}-C^{(j-1)}}{C^{(j)}}\big|\leq \epsilon$ \Comment{Verify convergence}
		
		\State \textbf{return} solution $ \{\mathbf{z}^r,\mathbf{x}^r,\mathbf{q}, \mathbf{P_{\mathrm{sar}}}, \mathbf{P_{\mathrm{com}}}\}$
		
	\end{algorithmic}
\end{algorithm}

To solve $\mathrm{(P.2)}$, we use SCA such that in each iteration $j$ of \textbf{Algorithm} \ref{algorithm1}, we update the solution $\{\mathbf{z}^{(j)},\mathbf{x}^{(j)}\}$ by solving problem $\mathrm{(P.3)}$ using conventional solvers such as CVX \cite{b10}.  To solve problem $\rm (P.1)$, we first solve problem $\rm(P.2)$ for
every fixed and feasible number of drone scans $N$, then we select the minimum number of azimuth scans that results in the maximum coverage to save drone energy. In the following proposition, we show that the number of feasible $N$ is finite, and therefore, the exhaustive search is guaranteed to yield the optimal number of azimuth scans $N$. 
\begin{proposition}
	\label{prop:feasible}
	The feasible number of azimuth scans $N$ is finite. The optimal number of azimuth scans $N^*$ is upper bounded by $ \frac{1}{M} \left(\frac{q_{\mathrm{start}}}{\delta_t P_{\mathrm{prop}}}+1\right)$. 
\end{proposition}
\begin{IEEEproof}
	Please refer to Appendix \ref{ap:proposition3}.
\end{IEEEproof}
\textbf{Algorithm} \ref{algorithm1} has a polynomial time complexity \cite{dinh2010local}. The exhaustive search is of constant time complexity as the  total number of iterations is less than $ \Big \lceil \frac{1}{M} \left(\frac{q_{\mathrm{start}}}{\delta_t P_{\mathrm{prop}}}+1\right)\Big\rceil$ according to \textbf{Proposition} \ref{prop:feasible}. Therefore, the overall time complexity of the proposed robust scheme is still polynomial. Moreover,  \textbf{Algorithm} \ref{algorithm1} converges to a local optimum of problem $\rm (P.1)$ as \ac{sca} methods are known to converge fast to locally optimal solutions\cite{complexity1}.
{Next, to assess the performance of the proposed robust sub-optimal scheme, we derive an upper bound on the optimal solution of problem $\rm (P.1)$ using monotonic optimization theory \cite{response2}}. 
\ifonecolumn
\begin{figure}[] 
	\centering
	\begin{tabular}{cc}
		\includegraphics[width=0.4\columnwidth]{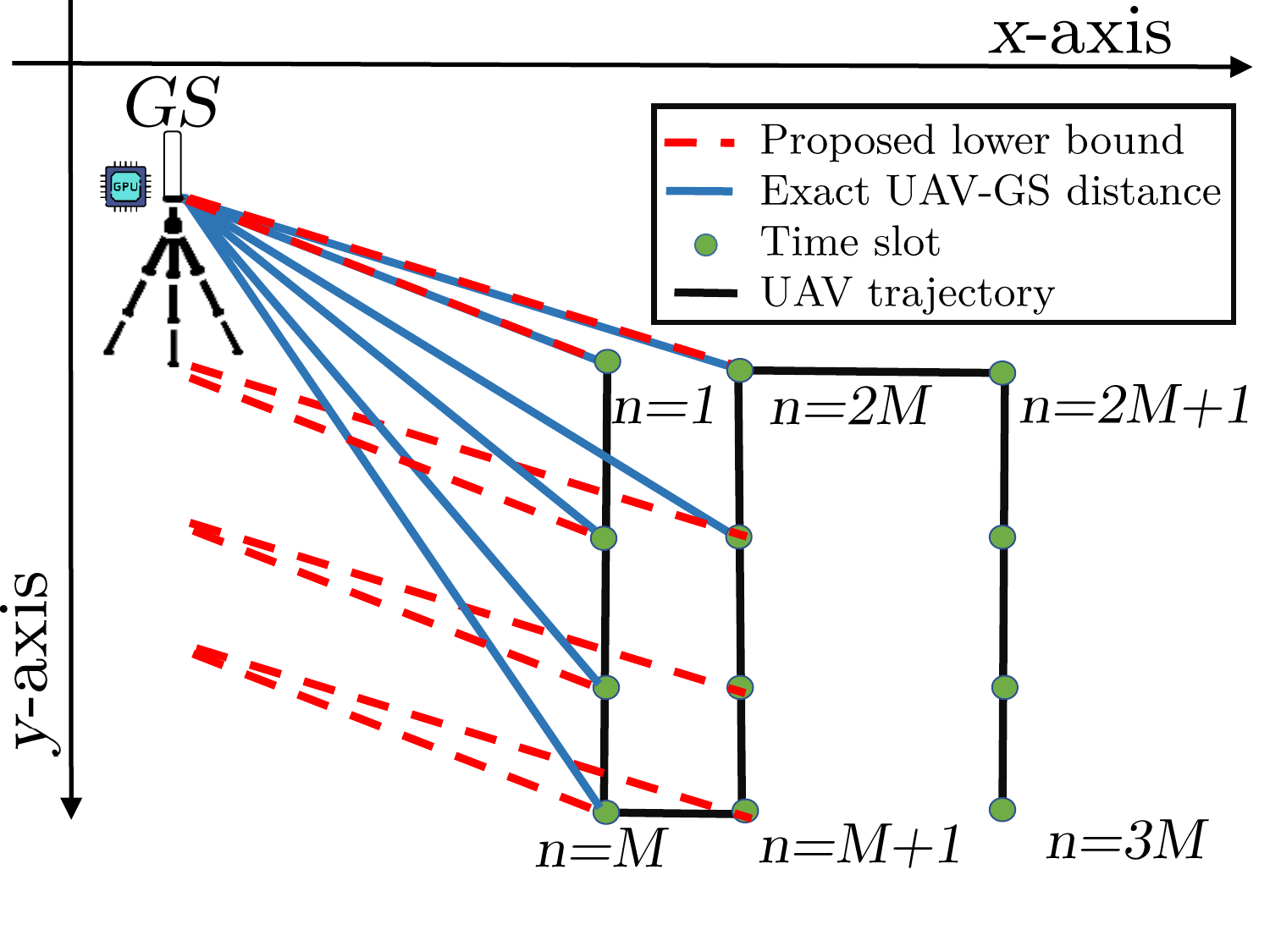}  & 
		\includegraphics[width=0.4\columnwidth]{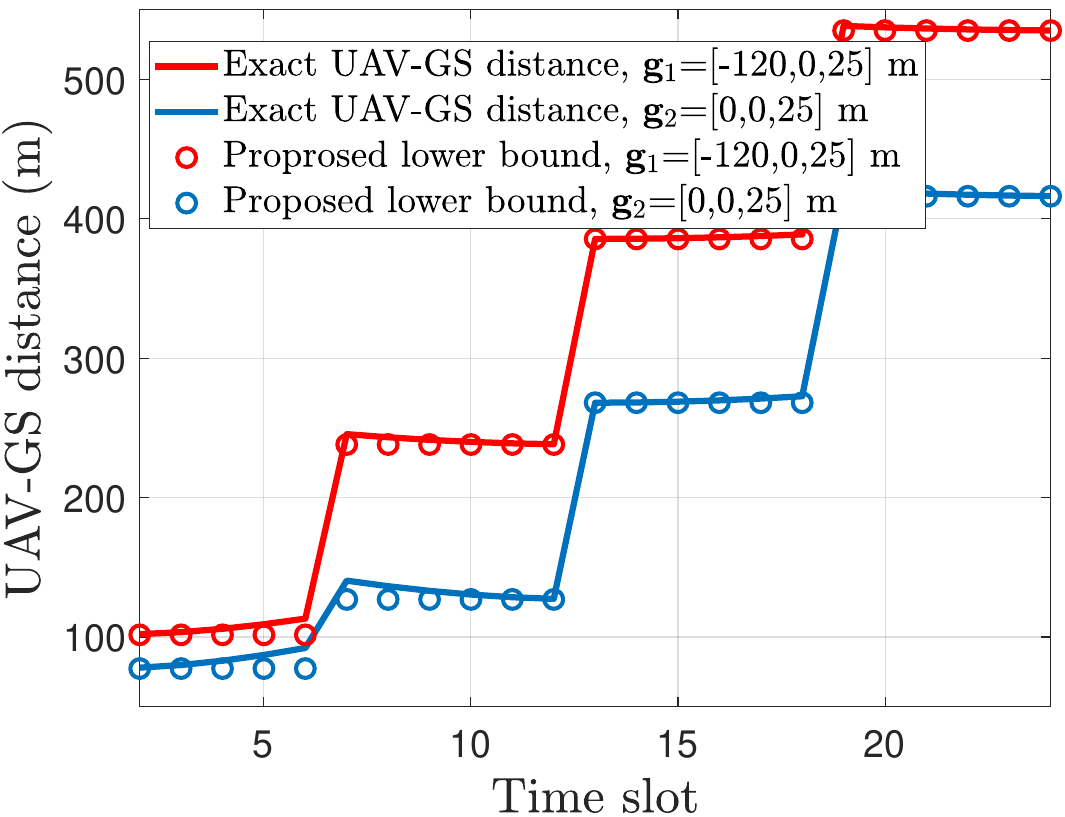}\\
		\textbf{(a)}& \textbf{(b)}  \\
	\end{tabular}
	\caption{ Lower bound on the \ac{uav}-\ac{gs} distance. (a) Proposed lower bound. (b) Accuracy of the proposed lower bound.}
	\label{fig:Upperbound}
\end{figure} 
\fi 
\subsection{Upper Bound on the Optimal Solution of Problem $\mathrm{(P.1)}$}
To provide an upper bound to problem $\mathrm{(P.1)}$, we start by fixing the total number of azimuth scans $N$, which leads to problem  $\mathrm{(P.2)}$. However, problem $\rm (P.2)$ is  non-monotonic w.r.t. optimization variables $\mathbf{x}$,  $\mathbf{z}$, and $\mathbf{P_{\rm com}}$, due to constraints $\rm C2$ and $\rm C8$. Equality constraints $\rm C3-C5$ also represent  an obstacle to applying monotonic optimization theory.   In what follows, we provide an upper bound on the optimal solution to problem $\rm (P.2)$. To this end, we introduce an optimization problem, denoted by $\mathrm{(\overline{P.2})}$, the optimal solution of which presents an upper bound on the optimal solution to problem $(\mathrm{P.2})$. Then, we use monotonic optimization theory to provide the optimal solution to problem $\mathrm{(\overline{P.2})}$. 
\subsubsection{Upper Bound Formulation}
During the $M$ time slots of a given azimuth scan, the \ac{uav}-\ac{gs} distance slightly varies depending on the \ac{aoi}-\ac{gs} distance. In each azimuth scan, we impose an upper bound on the \ac{uav} throughput $R_n$ by lower-bounding the \ac{uav}-\ac{gs} distance $d_n$. In fact, for a given azimuth scan, we consider the minimum possible \ac{uav}-\ac{gs} distance in the considered scan, see Figure \ref{fig:Upperbound}a. Following the proposed lower-bound, the \ac{uav}-\ac{gs} distance is constant during a scan and only varies from scan to scan. This approach allows us to focus on optimizing the system parameters from scan to scan, instead of optimizing the system along all time slots, which reduces the problem dimension from $N \times M$ to $N$. Hereinafter, each $n \in \mathbb{N}_{N}$, refers to the start of a different azimuth scan. The $y$-positions that correspond to the minimum \ac{uav}-\ac{gs} distance for each azimuth scan, denoted by $\overline{\mathbf{y}} = [\overline{y}(1), ..., \overline{y}(N)]^T\in \mathbb{R}^{N\times1}$, are given by: 
\begin{equation}
\overline{y}(n)=y(i), i=\argmin\limits_{ 1+(n-1)M\leq k \leq  nM}(y(k)-g_y)^2. \label{eq:lowerboundy}
\end{equation}
In Figure \ref{fig:Upperbound}b, we plot the exact \ac{uav}-\ac{gs} distance with solid lines and  the proposed lower bound on the \ac{uav}-\ac{gs} distance with circle markers. The figure confirms the tightness of the proposed lower bound for different placements of the \ac{gs}. The parameter values used in Figure \ref{fig:Upperbound}b are specified in Table \ref{Tab:System}. As each time slot now {  corresponds} to a different azimuth scan, we define vectors $\overline{\boldsymbol{\delta}_x^r}$ and  $\overline{\boldsymbol{\delta}_z^r}$ that contain the corrections in $x$- and $z$-direction, respectively, for each of these scans as follows: 
\begin{gather}
	\overline{\boldsymbol{\delta}^r_x}=\left[\delta_x^r(1),\delta_x^r(1+M),...,\delta_x^r(1+(N-1)M)\right]^T \in \mathbb{R}^{N\times1}, \\
	\overline{\boldsymbol{\delta}_z^r}=\left[\delta_z^r(1),\delta_z^r(1+M),...,\delta_z^r(1+(N-1)M)\right]^T \in \mathbb{R}^{N\times1},
\end{gather}
Based on this approach, the following $N$-dimensional optimization problem denoted by $\rm (\overline{P.2})$ is obtained as an upper bound for optimization problem $\rm (P.2)$: 
\begin{alignat*}{2}         
	&\!(\overline{\mathrm{P.2}}):\max_{\mathbf{z},\mathbf{q}, \mathbf{P_{\mathrm{sar}}},\mathbf{P_{\mathrm{com}}}} \hspace{3mm}\frac{L}{\Delta_s}C(\overline{\mathbf{u}^{r}}) & \qquad&  \\
	&\text{s.t.} &   & \; \\
	& \overline{\mathrm{C6}}: \;z_{\mathrm{min}}\leq \overline{z^{r}}(n) \leq z_{\mathrm{max}},  \forall n \in \mathbb{N}_N,                &      &  \\&  \overline{\mathrm{C7}}:\mathrm{SNR}_n(\mathbf{P}_{\rm sar},\overline{\mathbf{u}^{r}}) \geq \mathrm{SNR_{min}}, \forall  n \in \mathbb{N}_N,                &      &     
	\\
	&  \overline{\mathrm{C8}}:  R_n(\mathbf{P}_{\rm com}, \overline{\mathbf{u}^{r}}) \geq R_{\mathrm{min},n}(\overline{\mathbf{u}^{r}})+R_{\mathrm{sl}},  \forall n \in \mathbb{N}_N,                  &      &     
	\\
	&    \overline{\mathrm{C9}}: \;0\leq P_{\mathrm{sar}}(n)\leq P_{\mathrm{sar}}^{\mathrm{max}}, 0\leq P_{\mathrm{com}}(n)\leq P_{\mathrm{com}}^{\mathrm{max}},  \forall n \in \mathbb{N}_N,              &      & \\
	& \overline{\mathrm{C10}}: q(1)= q_{\mathrm{start}},\;q(n) \geq 0,  \forall n \in \mathbb{N}_N,       & & \\
	& \overline{\mathrm{C11}}: q(n+1) = q(n) - M\delta_t P_{\mathrm{tot}}(n),  \forall n \in { \mathbb{N}_{N-1}},       & &  
\end{alignat*}
where $\overline{z^{r}}(n)=z(n)+\overline{\delta_z^{r}}(n)$,
and the \ac{uav} position in each azimuth scan $n$, denoted by $\overline{\mathbf{u}^{r}}(n)$, is given by:
\begin{align}
\overline{\mathbf{u}^{r}}(n)=\left[l^r_n(\mathbf{z}),\overline{y}(n),\overline{z^{r}}(n)\right]^T,  \forall n \in \mathbb{N}_N,
\end{align}
where $l^{r}_n(\mathbf{z})=(c_2-c_1)\sum\limits_{k=1}^{n-1} z(k)-c_1\;z(n)+\overline{\delta_x^{r}}(n)$. Problem $\mathrm{(\overline{P.2})}$ is an $N$-dimensional non-convex non-monotonic \ac{minlp}. Next, we show that $(\overline{\mathrm{P.2}})$ is an upper bound to problem $(\mathrm{P.2})$, then, we provide the optimal solution of $\rm \overline{(P.2)}$ using monotonic optimization.

 \begin{theorem}\label{theo:P.2}
	Let $\overline{\mathbf{u}^{r}}$ and $\mathbf{u}^{r}$ be the optimal solutions to problems $(\overline{\mathrm{P.2}})$ and $\mathrm{(P.2)}$, respectively, then $\frac{L}{\Delta_s}C(\overline{\mathbf{u}^{r}})\geq C(\mathbf{u}^{r}), \forall r \in [0,1].$ In other words, the optimal solution to problem $(\overline{\mathrm{P.2}})$ provides an upper bound for the optimal solution to problem $\mathrm{(P.2)}$.
\end{theorem}
\begin{proof}
	Please refer to Appendix \ref{app:P.2}.
\end{proof}
\subsubsection{Reformulation as Monotonic Optimization Problem}
To provide the optimal solution to problem $\rm (\overline{P.2})$, we write the problem in the canonical form of a monotonic optimization problem given by \cite{1,2}: 
\begin{equation}\label{eq:canonicalform}
	\max \{ \mathcal{C}(\mathbf{x}) | \mathbf{x} \in \mathcal{G} \cap \mathcal{H}   \},
\end{equation}
where $\mathcal{C}: \mathbb{R}^{N \times 1}_+ \rightarrow \mathbb{R}$ is the cost function which is increasing w.r.t. $\mathbf{x}\in \mathbb{R}^{N\times1}$, $\mathcal{G} \subset [\mathbf{0},\; \mathbf{b}] \subset  \mathbb{R}^{N \times 1}_+$ is a compact normal set with non-empty interior, and $\mathcal{H}$ is a conormal set on the box $[\mathbf{0}, \; \mathbf{b}]$\footnote{Rigorous definitions of increasing functions, boxes,  normal sets, and conormal sets are provided in the monotonic optimization theory guides \cite{1,2}.}. 
As a first step, we rewrite constraint $\overline{{\rm C8}}$ as follows: 
\begin{equation} \label{eq:simplified}
	\overline{{\rm C8}}: \left( A2^{\alpha \overline{z^{r}}(n)} -1\right)d^2_n(\overline{\mathbf{u}^{r}})-P_{\mathrm{com}}(n)\gamma\leq 0,  \forall n \in \mathbb{N}_N.
\end{equation}
Constraint $\overline{\mathrm{C8}}$ is non-convex and non-monotonic w.r.t. optimization variable $\mathbf{z}$ and $\mathbf{P_{\rm com}}$. Thus, we transform $\overline{\mathrm{C8}}$ into a set of monotonic constraints by first writing it as a difference of monotonic functions.
\begin{proposition}
	Constraint $\overline{\rm C8}$ is a difference of two monotonic functions  $f(z(n)) -g(z(n),P_{\rm com}(n))$, where $f(z(n))$ is given by:
	\ifonecolumn
	\begin{gather}
		f(z(n))=\left(2^{\alpha \overline{z^{r}}(n)}-1\right) \Bigg((c_2-c_1)^2 \left(\sum_{k=1}^{n-1}z(k)\right)^2  + 2z(n) \Big(c_1[g_x]^+-c_1\overline{\delta_x^{r}}(n) +\overline{\delta_z^{r}}(n)\Big) +\notag\\ \left(1+{c_1}^2\right) z^2(n)+ 2(c_2-c_1)\sum_{k=1}^{n-1}z(k) [-g_x]^++ \left(g_x-\overline{\delta_x^{r}}(n)\right)^2 +\left(g_z-\overline{\delta_z^{r}}(n)\right)^2 +\left(\overline{y}(n)-g_y\right)^2\Bigg),
	\end{gather}
		and $g(z(n),P_{\rm com}(n))$ is given by :
		\begin{gather} 
			g(z(n),P_{\rm com}(n))= \gamma P_{\rm com}(n)+\left(2^{\alpha \overline{z^{r}}(n)}-1\right) \times  \Bigg(2(c_1[-g_x]^++g_z) z(n)   + \notag \\ (c_2-c_1)\left(2[g_x]^++c_1\;z(n)\right)\sum_{k=1}^{n-1}z(k)\Bigg), n \in \mathbb{N}_N. \label{eq:functiong}
		\end{gather}
	\fi
\end{proposition}
\begin{proof}
	The result can be easily obtained by
	verifying that functions $f$ and $g$ are increasing  \ac{wrt} $z$ and $f(z(n)) -g(z(n),P_{\rm com}(n))$ is equal to the left-hand side of (\ref{eq:simplified}).
\end{proof}
Constraint $\overline{\rm C8}$  is a difference of monotonic functions $f$ and $g$, thereby, as proven in \cite{1},  we can equivalently  represent it with the following constraints:  
\begin{align} \label{eq:3constraints}
	&\overline{\mathrm{ C8a}}: f(z(n)) +t(n)\leq f(z_{\rm max}),  \forall n \in \mathbb{N}_N,\\
	&\overline{\mathrm{ C8b}}:    g(z(n),P_{\rm com}(n)) +t(n) \geq f(z_{\rm max}),  \forall n \in \mathbb{N}_N,\\
	&\overline{\mathrm{ C8c}}:    0 \leq t(n) \leq f(z_{\rm max})-f(0),  \forall n \in \mathbb{N}_N,
\end{align}
where $\mathbf{t} =[t(1), ...,t(N)]^T\in \mathbb{R}^{N \times 1}$ is a vector of auxiliary optimization variables.  
Problem $(\overline{\rm P.2})$ can now be written in the canonical form of a monotonic optimization problem as follows:  
\begin{alignat*}{2}
	&\!(\overline{\rm P.3}):\max_{\mathbf{z}, \mathbf{t},\mathbf{q}, \mathbf{P_{\mathrm{sar}}}, \mathbf{P_{\mathrm{com}}}} \frac{L}{\Delta_s}C(\overline{\mathbf{u}^{r}})  & \qquad&  \\
	&\text{s.t.} \hspace{3mm}\mathrm{\overline{C6}, \overline{C7}, \overline{C8a}, \overline{C8b}, \overline{C8c}, \overline{C9}, \mathrm{\overline{ C10}}, \overline{C11}}, &   & \;  
\end{alignat*}
where the feasible set is the intersection of normal set $\mathcal{G}$ that is spanned by constraints ($\rm \overline{C6}$, $\rm \overline{C7}$,  $\rm\overline{C8a}$, $\rm\overline{C8c}$, $\rm \overline{C9}$, $\rm \overline{C10}$, $\rm \overline{C11}$) and conormal set $\mathcal{H}$ that is spanned by constraints $\rm (\overline{C6}, \overline{C8b})$. The optimal solution of problem $(\overline{\rm P.3})$ lies at the upper boundary of the normal set $\mathcal{G}$, denoted by $\partial^+\mathcal{G}$ \cite{2}. This region is not known a priori but the sequential  polyblock approximation algorithm can be used to approach this set from above. 
\subsubsection{Polyblock Approximation Algorithm}

\begin{algorithm}
	\caption{Polyblock Outer Approximation Algorithm }\label{algorithm2}
	\begin{algorithmic}[1] 
		\State Initialization: Set polyblock $\mathcal{P}_{1}$ with vertex $\mathbf{v}_{\rm max} = [\mathbf{z}_{\rm max},\mathbf{t}_{\rm max},\mathbf{P_{\rm com}^{\rm max }}]^T$, vertex set $\mathcal{T}^{(1)}=\{ \mathbf{v}_{\rm max}\}$, small positive number $\epsilon\geq0$, objective function value $ \rm CBV_0=- \infty$, and iteration number $k=0$.
		\State \textbf{repeat}
		\State\hspace{\algorithmicindent} $ k=k+1$.\Comment{{ Increment iteration index}}
		\State\hspace{\algorithmicindent} Select $ \mathbf{v}_k=\argmax_{ \mathbf{v} \in \mathcal{T}^{(k)}}(C(\mathbf{v}))$\Comment{{Select the best vertex in set $\mathcal{T}^{(k)}$ }}
		\State\hspace{\algorithmicindent}  Compute the projection $\Phi(\mathbf{v}_{k})$ on $\partial \mathcal{G}^+$ using \textbf{Algorithm} \ref{algorithm3}
		\State \hspace{\algorithmicindent}\textbf{if} $\Phi(\mathbf{v}_{k})=\mathbf{v}_{k}$, i.e., $\mathbf{v}_{k} \in \mathcal{G}$ \textbf{then}
		\State \hspace{\algorithmicindent} \hspace{\algorithmicindent}  $\mathbf{s}_{k}=\mathbf{v}_{k}$ and $\mathrm{CBV}_{k}=C(\Phi(\mathbf{v}_{k}))$
		\State \hspace{\algorithmicindent}\textbf{else}
		\State \hspace{\algorithmicindent}\hspace{\algorithmicindent} \textbf{if}  $\Phi(\mathbf{v}_{k}) \in \mathcal{G} \cap \mathcal{H}$ and $C(\Phi(\mathbf{v}_{k}))\geq \mathrm{CBV}_{k-1}$  \textbf{then} \Comment{{ Better solution is found}}
		\State \hspace{\algorithmicindent}\hspace{\algorithmicindent} \hspace{\algorithmicindent} Let the current best solution $\textbf{s}_{k}=\Phi(\mathbf{v}_{k})$ and $\mathrm{CBV}_{k}=C(\Phi(\mathbf{v}_{k}))$ \Comment{{ Update best solution}}
		\State  \hspace{\algorithmicindent} \hspace{\algorithmicindent}  \textbf{else} 
		\State\hspace{\algorithmicindent} \hspace{\algorithmicindent} \hspace{\algorithmicindent} $\textbf{s}_{k}=\textbf{s}_{k-1}$ and $\mathrm{CBV}_{k}=\mathrm{CBV}_{k-1}$
		\State\hspace{\algorithmicindent} \hspace{\algorithmicindent} \textbf{endif}
		\State\hspace{\algorithmicindent}\hspace{\algorithmicindent} Set $\mathcal{T}^{(k+1)}=\{\mathcal{T}^{(k)}\setminus \mathcal{T}^*\}\cup \{ \mathbf{v}^{i}=\mathbf{v}+({(\Phi(\mathbf{v}_k))}_i-v_i)\mathbf{e}_{i} | \mathbf{v} \in \mathcal{T}^{*},i\in \{1,..,3N\} \}$,
			\Statex\hspace{\algorithmicindent}\hspace{\algorithmicindent}  where  $\mathcal{T}^{*}= \{ \mathbf{v} \in \mathcal{T}^{(k)} | \mathbf{v} > \Phi(\mathbf{v}_{k})\}$
		\State\hspace{\algorithmicindent} \hspace{\algorithmicindent} Remove from $\mathcal{T}^{(k)} $ improper vertices and vertices $\{ \mathbf{v} \in\mathcal{T}^{(k+1)}|  \mathbf{v} \notin \mathcal{H}\}$ \Comment{{ Discard unnecessary vertices from set $\mathcal{T}^{(k+1)}$}}
		\State\hspace{\algorithmicindent} \textbf{endif}
		\State\hspace{\algorithmicindent} \textbf{until} $|C(\mathbf{v}_{k})-\mathrm{CBV}_k|\leq \epsilon$\Comment{{Verify convergence}}
		\State\hspace{\algorithmicindent} \textbf{return} solution $\mathbf{s}_k$ 
	\end{algorithmic}
\end{algorithm}

In \textbf{Algorithm} \ref{algorithm2}, we provide all steps of the proposed polyblock outer approximation algorithm that generates an $\epsilon$-approximate  optimal solution to problem $\rm (\overline{P.3})$. In addition, Figure \ref{fig:MO}  depicts the operation of the polyblock algorithm where, for simplicity of illustration, we consider a case with only two dimensions
$x$ and $y$. The optimization variables of problem $(\overline{\rm P.3})$  are divided into outer optimization variables, that are collected in vertex vector $\mathbf{v}=\big[{\overline{\mathbf{z}^{r}}\;}^T$, $\mathbf{t}^T$, ${(\mathbf{P}_{\rm com})}^T\big]^T\in \mathbb{R}^{3N\times1}$, and inner optimization variables, $\mathbf{q}$ and  $\mathbf{P_{\mathrm{sar}}}$. Hereinafter, as the objective function depends only on $ \mathbf{z}$, for ease of notation, $C(\mathbf{v})$  refers to $ \frac{L}{\Delta_s} C([v(1),...,v(N)]^T)$.

\begin{figure*}[] 
	\centering
	\begin{tabular}{ccccc}
		\includegraphics[width=0.23\linewidth]{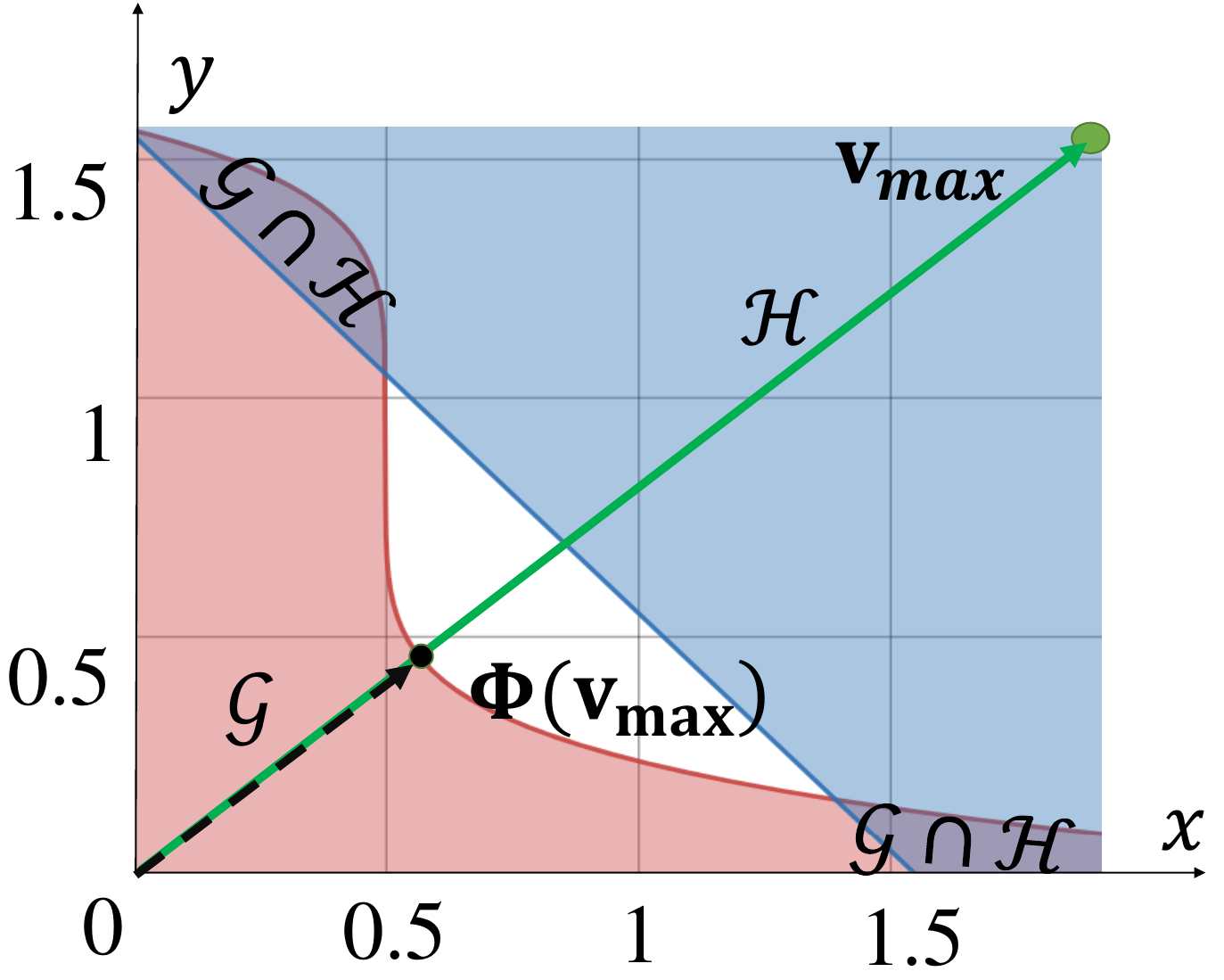}  &
		\includegraphics[width=0.23\linewidth]{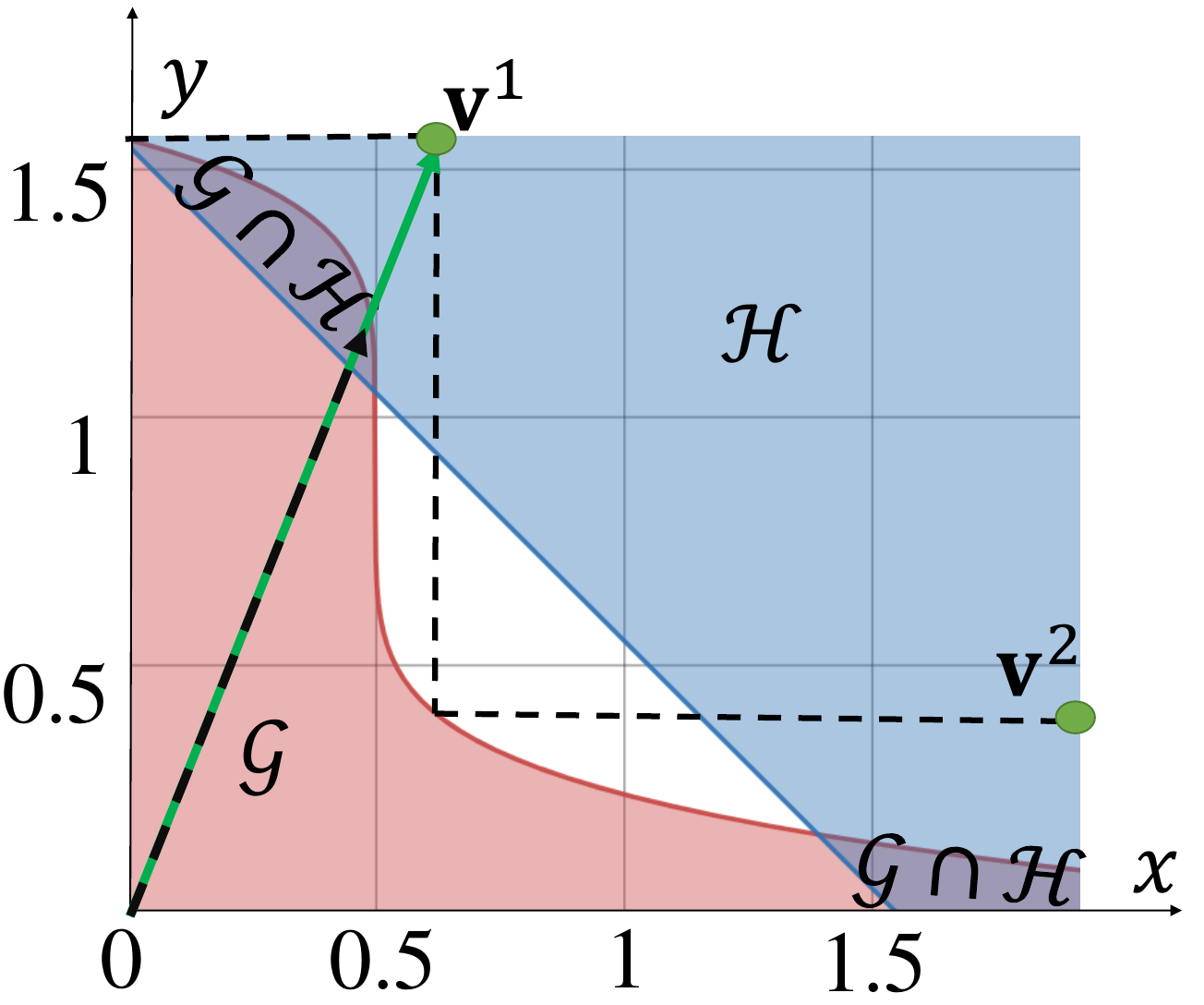}&
		\includegraphics[width=0.23\linewidth]{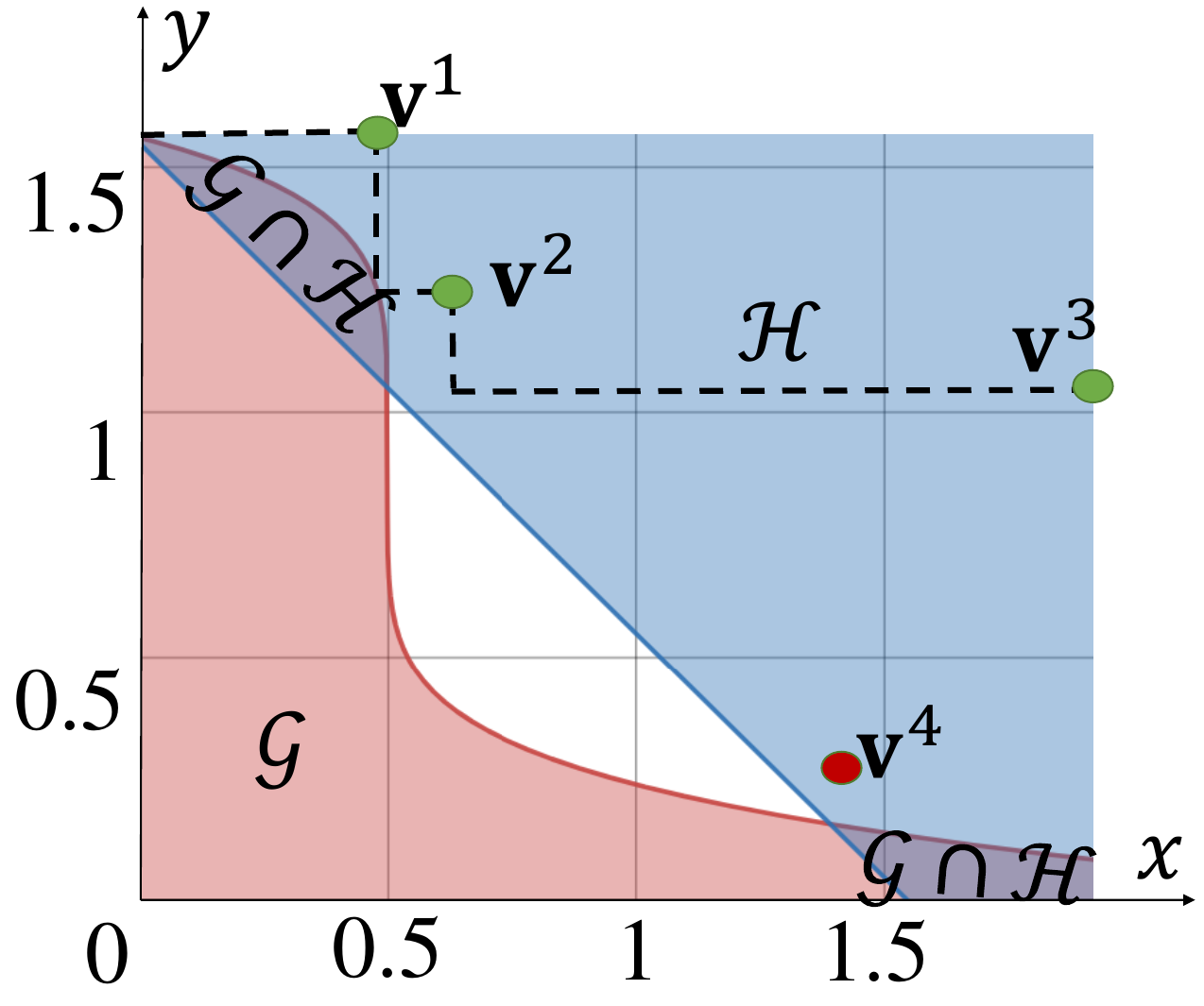}&
		\includegraphics[width=0.23\linewidth]{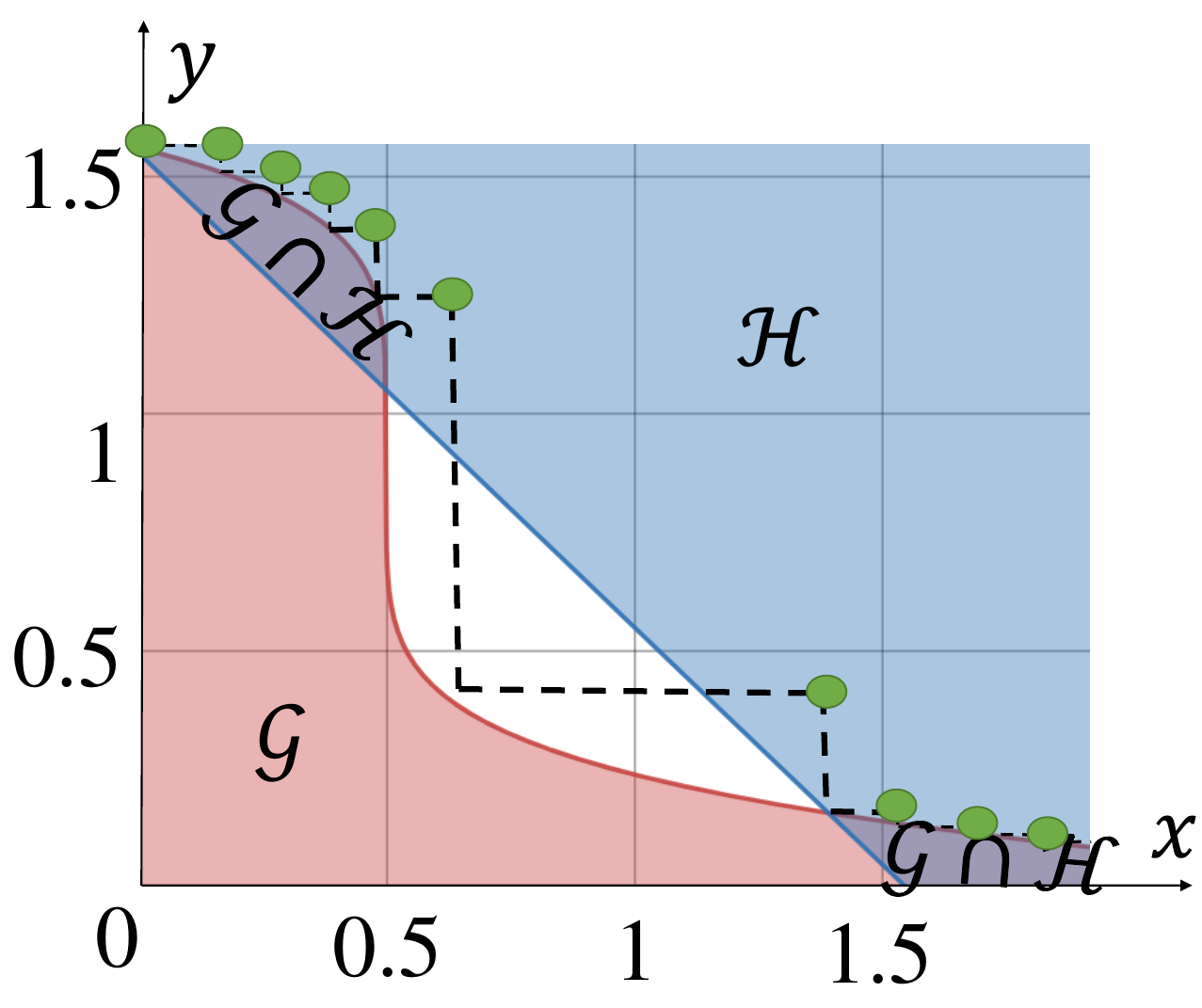} \\
		\textbf{(a)}& \textbf{(b)}& \textbf{(c)} &\textbf{(d)}  \\
	\end{tabular}
	\caption{Illustration of polyblock outer approximation algorithm. (a) Feasible set for a monotonic maximization problem where the normal and conormal sets are denoted by $\mathcal{G}$ and $\mathcal{H}$, respectively. (b) Projection of $\mathbf{v}_{\rm max}$ on normal set $\mathcal{G}$. (c) Projection of the best vertex $\mathbf{v}^1=\argmax\limits_{\mathbf{v} \in \mathcal{T}^{(1)}}\{ C(\mathbf{v})\}$ on normal set $\mathcal{G}$ and removal of improper vertex $\mathbf{v}^4$ (dominated by $\mathbf{v}^3$). (d)  Representation of a tighter polyblock obtained after several iterations.  }
	\label{fig:MO}
\end{figure*}  
As shown in Figure \ref{fig:MO}a, first, we { initialize} a polyblock $\mathcal{P}^{(1)}$ with corresponding vertex set $\mathcal{T}^{(1)}$, which initially contains vertex $\mathbf{v}_{\rm max} =[\mathbf{z}_{\rm max},\mathbf{t}_{\rm max},\mathbf{P_{\rm com}^{\rm max }}]^T \in \mathbb{R}^{3N\times 1}$, where $\mathbf{z}_{\rm max}=[z_{\rm max}-\overline{\delta_z^{r}}(1),...,z_{\rm max}-\overline{\delta_z^{r}}(N)]\in \mathbb{R}^{1\times N}$, $\mathbf{t}_{\rm max}=[f(z_{\rm max})-f(0), ...,f(z_{\rm max})-f(0)]\in \mathbb{R}^{1\times N}$, and $\mathbf{P_{\rm com}^{\rm max }}=[P_{\rm com}^{\rm max }, ...,P_{\rm com}^{\rm max } ]\in \mathbb{R}^{1\times N}$.   The projection of the vertex $\mathbf{v}_{\rm max}$ onto the normal set $\mathcal{G}$, denoted by $\Phi(\mathbf{v}_{\rm max})\in \mathbb{R}^{3N \times 1}$, is calculated and used to construct a tighter polyblock $\mathcal{P}^{(2)}$, see Figure  \ref{fig:MO}b. The new vertex set is  calculated as  $\mathcal{T}^{(2)}=\{\mathcal{T}^{(1)}\setminus \mathbf{v}_{\rm max}\} \cup \{\mathbf{v}^{i}=\mathbf{v}_{\rm max}+\left( {\left(\Phi( \mathbf{v}_{\rm max}) \right)}_i-{(\mathbf{v}_{\rm max})}_{i}\right)\mathbf{e}_i, i \in \{1, ..., 3N\}\}$, where ${\left(\Phi( \mathbf{v}_{\rm max}) \right)}_i$ and ${(\mathbf{v}_{\rm max})}_{i}$ are the $i^{th}$-elements of vectors  $\Phi(\mathbf{v}^{\rm max})$ and $\mathbf{v}_{\rm max}$, respectively.  Here, $\mathbf{e}_i$  is the unit vector whose $i^{th}$ entry is equal to $1$ and all other entries are equal to $0$. Then, in iteration $k$, the vertex that maximizes the objective function of problem $\mathrm{(\overline{P.3})}$ is selected to create the next polyblock, i.e., $\mathbf{v}_k=\argmax\limits_{\mathbf{v} \in \mathcal{T}^{(k)}}\{ C(\mathbf{v})\}$. Once a feasible projection (i.e., $\Phi(\mathbf{v}_k) \in \mathcal{G}\cap\mathcal{H}$) that improves the objective function of problem $\rm (\overline{P.3})$ is located, the corresponding vector $\mathbf{s}_k=\Phi(\mathbf{v}_k)$ and its objective function  $\mathrm{CBV}_k=C(\Phi(\mathbf{v}_k))$ are recorded.  To speed up convergence and save memory, improper vertices\footnote{Given a vertex set $\mathcal{T},$ improper vertices are vertices that are dominated by other vertices in  $ \mathcal{T}.$ See Figure \ref{fig:MO}c and \cite{1} for further details.} and vertices that are not in set $\mathcal{H}$ are removed in each iteration, e.g., vertex $\textbf{v}^4$ in Figure \ref{fig:MO}c.
The same procedure is repeated multiple times, such that a sequence of polyblocks converges from above to the feasible set, i.e., $\mathcal{P}^{(1)}\supset\mathcal{P}^{(2)}\supset...\supset\mathcal{P}^{(k)}\supset \mathcal{G}\cap\mathcal{H}$, see Figure \ref{fig:MO}d. The termination criterion of the proposed algorithm  is  $|C(\mathbf{v}_k)-\mathrm{CBV}_k|\leq \epsilon$, where $\epsilon$ is the error tolerance.

\begin{algorithm}
	\caption{Bisection Search to Compute the Upper Boundary Point }\label{algorithm3}
	\begin{algorithmic}[1] 
		\State Initialization: Set vertex $\mathbf{v}$, $ \lambda_{\rm min}=0$, $ \lambda_{\rm max}=1$, and error tolerance $0< \epsilon \ll 1$.
		\State \textbf{repeat}
		\State\hspace{\algorithmicindent} $ \lambda= \frac{(\lambda_{\rm max}+\lambda_{\rm min})}{2}$ \Comment{{ Bisection middle point}}
		\State\hspace{\algorithmicindent}\textbf{if} $\lambda \mathbf{v} $ is feasible for problem $\mathrm{(\overline{P.3})}$  \textbf{then}$ $ \State\hspace{\algorithmicindent}\hspace{\algorithmicindent}Set $ \lambda_{\rm min}=\lambda$ \Comment{{ Search in [$\lambda,\lambda_{\rm max}$]}}
		\State\hspace{\algorithmicindent}\textbf{else}
		\State\hspace{\algorithmicindent}\hspace{\algorithmicindent}Set $\lambda_{\rm max}=\lambda$\Comment{{ Search in [$\lambda_{\rm min},\lambda$]}}\
		\State\hspace{\algorithmicindent}\textbf{end if}
		\State \textbf{until} $\lambda_{\rm max}-\lambda_{\rm min}\leq \epsilon$\Comment{{ Verify convergence}}
		\State \textbf{return} $\Phi(\mathbf{v})=\lambda\mathbf{v}$
	\end{algorithmic}
\end{algorithm}

In Line $4$ of \textbf{Algorithm} \ref{algorithm2}, we compute the projection of vertex $ \mathbf{v}_k=\argmax\limits_{ \mathbf{v} \in \mathcal{T}^{(k)}}C(\mathbf{v})$ on $ \partial \mathcal{G}^+ $, denoted by $\Phi(\mathbf{v}_{k})\in \mathbb{R}^{3N \times 1} $. This projection is not trivial and is equivalent to solving a one-dimensional optimization problem: $\max\limits_{0\leq \lambda \leq 1}\{\lambda  \mathbf{v}_k \in \mathcal{G}\}$. To compute $\Phi(\mathbf{v}_{k})$, we perform a bisection search using \textbf{Algorithm} \ref{algorithm3}. In each iteration of   \textbf{Algorithm} \ref{algorithm3}, checking if a vector belongs to the normal set $\mathcal{G}$ is equivalent  to solving problem $(\overline{\rm P.3})$ when fixing the outer optimization variables, i.e., vectors $\mathbf{z}$, $\mathbf{t}$, and $\mathbf{P}_{\rm com}$. This results in the following feasibility problem: 
\begin{alignat*}{2}
	&\!\mathrm{(\overline{P.4})}:\max_{\mathbf{q}, \mathbf{P_{\mathrm{sar}},}}\hspace{3mm} 1       & \qquad&  \\
	&\text{s.t.} \hspace{3mm} \mathrm{ \overline{C7}, \overline{C9}, \overline{C10}, \overline{C11}}. &   & \;
\end{alignat*}
Problem $\mathrm{(\overline{P.4})}$ is a convex optimization problem and a feasible solution can be obtained by, e.g., CVX \cite{b10}.  Once the projection point is computed, the new polyblock is generated as described previously and these steps are repeated until the optimal solution is located.\par
The computational complexity of the  upper bound on the optimal solution, provided by \textbf{Algorithm} \ref{algorithm2} and \textbf{Algorithm} \ref{algorithm3}, is exponential  in the number of vertex elements, i.e., the time complexity order is $O(2^{3N})$ \cite{2}. If problem $\rm \overline{(P.2)}$ is feasible, then \textbf{Algorithm} \ref{algorithm2} converges to its global optimal solution as the polyblock approximation algorithm has been proven to converge to an $\epsilon$-optimal solution \cite{2}. Despite  the high complexity of the derived upper bound, it can be used to assess the performance of the proposed low-complexity sub-optimal scheme. 
\ifonecolumn
\begin{table}[]
	\centering
	\caption{	 System parameters  \cite{table2,table3,response1}.}
		\begin{tabular}{?c|c?c|c?}
	\Xhline{3\arrayrulewidth}
	Parameter                & 	Value &	 Parameter      & 		Value  \\ \Xhline{3\arrayrulewidth}
	$L$    & 	60 m                          			&	$\mathbf{g}=[g_x,g_y,g_z]^T$&	  $[0,0,5]^T$ m  \\ \hline
	$M$         & 	100                               			&	$v$    &	 5 m/s     \\ \hline
	 $z_{\rm max}$ and $z_{\rm min}$ &	 100 m and  2 m    		&	 $q_{\rm start}$      & 	 19.44 Wh       \\ \hline
	 $P_{\rm sar}^{\rm max}$       & 	 46 dBm          &	$P_{\rm prop}$  &	 450 W   \\ \hline
	$P_{\rm com}^{\rm max}$& 	 40 dBm   	     	&	 $P_0$&	 79.86 W   \\ \hline
	$\theta_d$ and $\theta_{\rm 3db}$  &	 45$^o$  and 30$^o$    					&	$P_I$ & 	420.6 W\\ \hline  
	$\tau_p$ and $\mathrm{PRF}$        & 	 1 $\mu$s  and 100 Hz     						&	$W_u$ &	 56.5 N  	 \\ \hline
	$B_r$        &  	 100 MHz     				&	$U_{\mathrm{tip}}$&	 120 m/s \\ \hline 
	$B_c$	&  	 100 MHz 		&	$\rho$ &	 1.225 kg/m$^3 $\\ \hline 
	$f$            & 	2 GHz     						&	$A$ &	0.503 m$^2$\\ \hline 
	$\mathrm{SNR}_{\rm min}$ &	 20 dB  					&	$d_0$&	 0.6 \\\hline 
	$R_{\rm sl}$ &	 1 kbit 							&	 $s$&	 0.05 \\ \hline 
	$\gamma$      &	  20 dB    							&	 $\beta$    &   $10^4$ W$^{-1}$ m$^3 $\\ \Xhline{3\arrayrulewidth}
\end{tabular}
	\label{Tab:System}
\end{table}

\fi

\section{Simulation Results and Discussion } \label{Sec:SimulationResults}
In this section, we present simulation results for the proposed robust \ac{uav} \ac{3d} trajectory and resource allocation algorithm, where we adopted the parameter values provided in Table \ref{Tab:System}, unless specified otherwise. { The minimum \ac{uav} altitude $z_{\rm min}$ was set to 2 m to avoid crashes due to collisions with obstacles at very low altitudes.} For all robust schemes, we adopt the Gaussian model for the \ac{uav} trajectory deviations where in-range position errors follow $ \mathcal{N}( o_x= 1\; \rm{ m} , \sigma = 0.3\; \rm{ m})$ and altitude errors follow $ \mathcal{N}( o_z= -1\;\rm { m}, \sigma = 0.3\;\rm{ m} )$.
\subsection{Benchmark Schemes}
To evaluate the performance of the proposed robust scheme, we compare it with the following
benchmark schemes:
\begin{itemize}
	\item \textbf{Benchmark scheme 1}: This is a non-robust scheme, where we solve problem $\rm (P.1)$ assuming no \ac{uav} trajectory deviations, i.e., $\sigma =0$, $o_x=0$, and $o_z=0$. This benchmark scheme was presented in \cite{14} where $\mathbf{u}^r=\mathbf{u}$.
	
	\item \textbf{Benchmark scheme 2}: This is a robust scheme, where we solve problem $\rm (P.1)$ assuming the communication power $P_{\rm com}$ is fixed across all time slots and optimized
	along with the other optimization variables. The coverage reliability level is $r=0.95$. 
	
	\item \textbf{Benchmark scheme 3}: 
	This is also a robust scheme, with $r=0.95$,  where we solve problem $\rm (P.1)$  assuming the \ac{sar} power $P_{\rm sar}$ is fixed across all time slots but optimized along with the other optimization variables.  
\end{itemize}
{ We note that all benchmark schemes have polynomial time complexity, similar to the proposed sub-optimal scheme. }
\subsection{Convergence and Accuracy of the Proposed Scheme}
\ifonecolumn
\begin{figure*}[] 
	\centering
	\begin{tabular}{ccc}
	\includegraphics[width=0.42\linewidth]{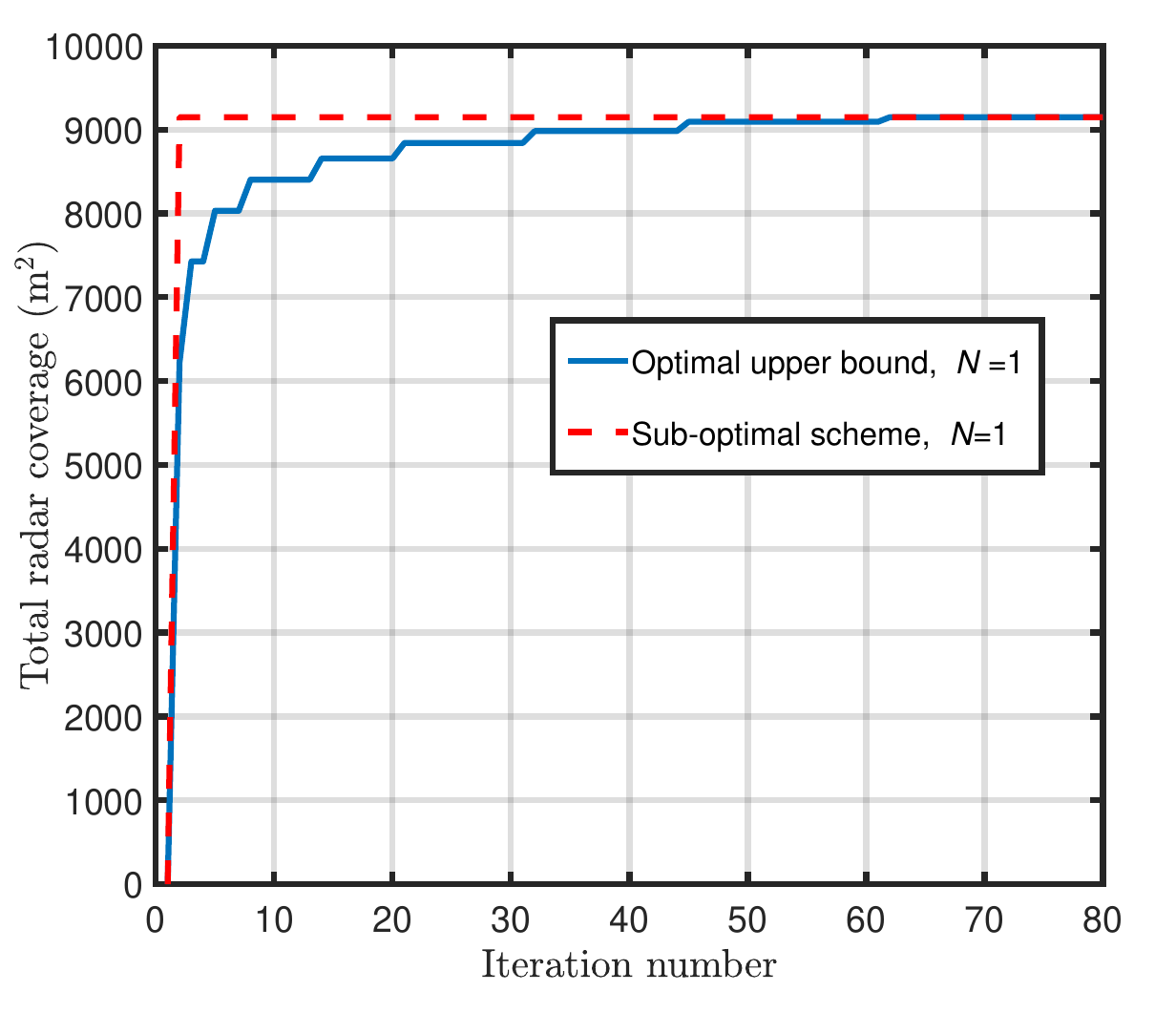} 	 &\includegraphics[width=0.42\linewidth]{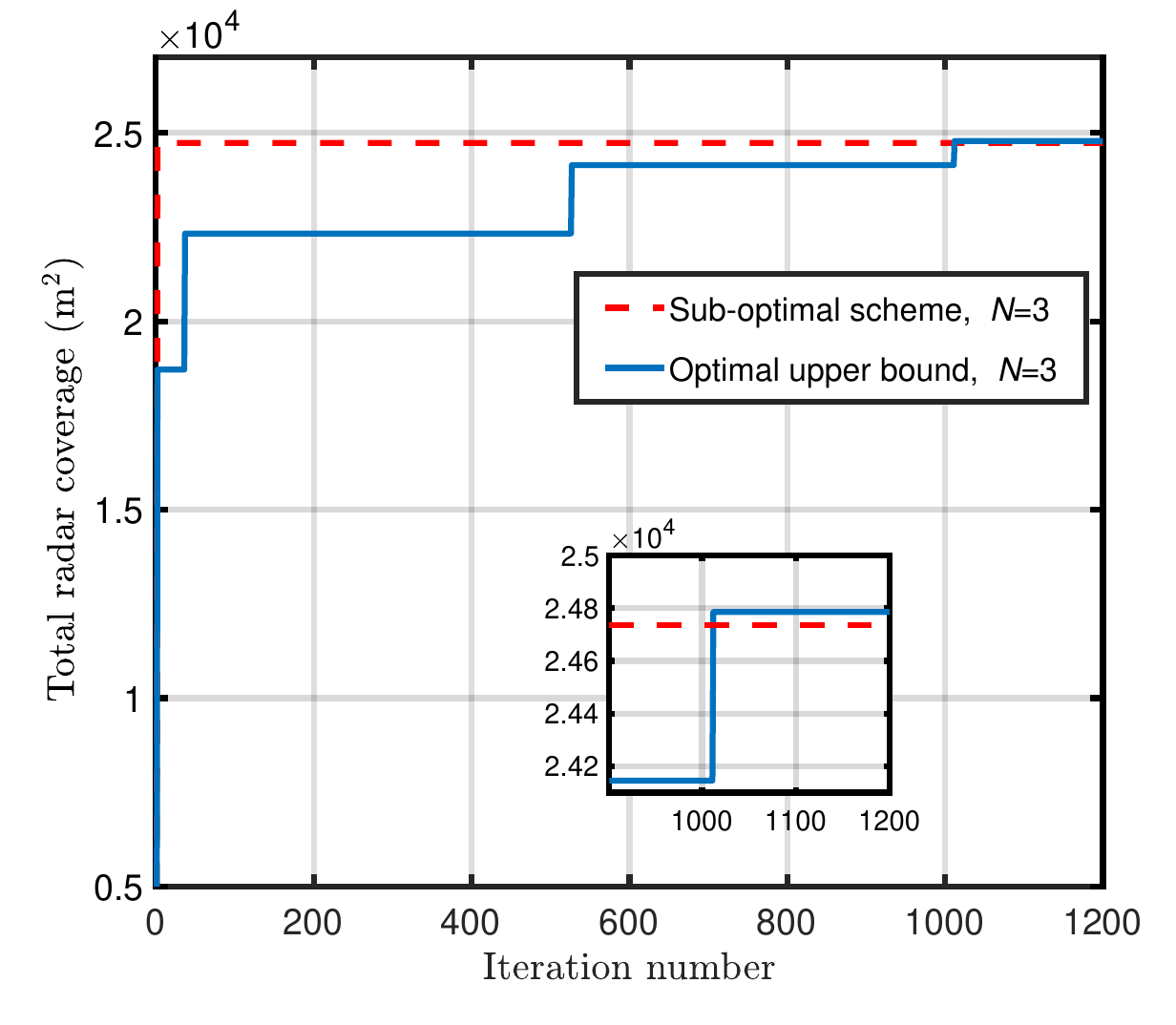} 
		\\
		\textbf{(a)} & \textbf{(b)} 
	\end{tabular}
	\caption{Convergence of algorithm 1 (sub-optimal) and algorithm 2 (upper bound) for $N=1$ (a) and $N=3$ (b). }
	\label{fig:Convergence}
\end{figure*} 
\fi

\ifonecolumn
\begin{figure}[]
	\centering
	\includegraphics[width=3.5in]{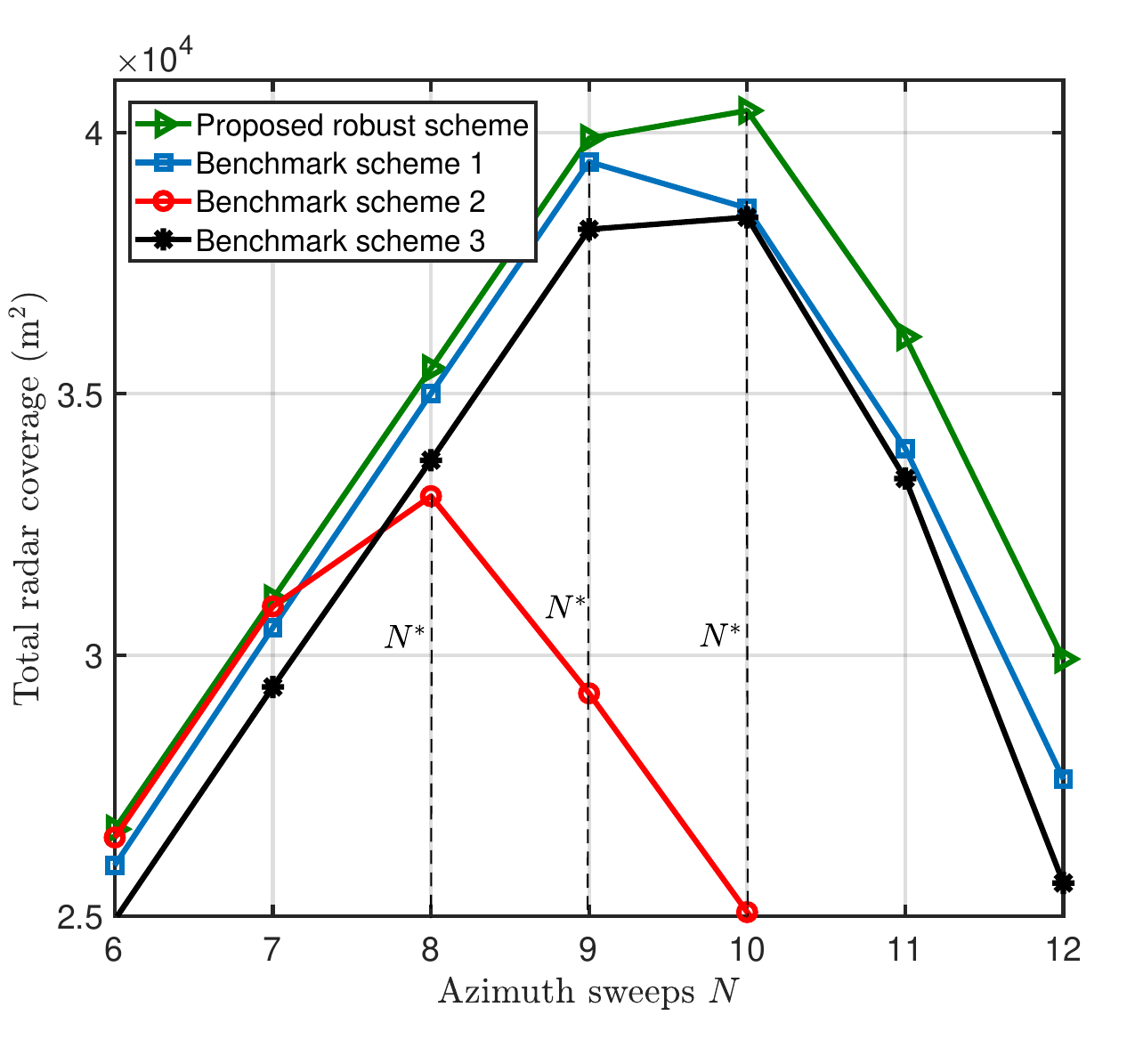}
\caption{	Comparison of the radar coverage of the considered schemes.}
	\label{fig:schemecomparison1} 
\end{figure}
\fi

\ifonecolumn
\begin{figure}[]
	\centering
	\includegraphics[width=3.5in]{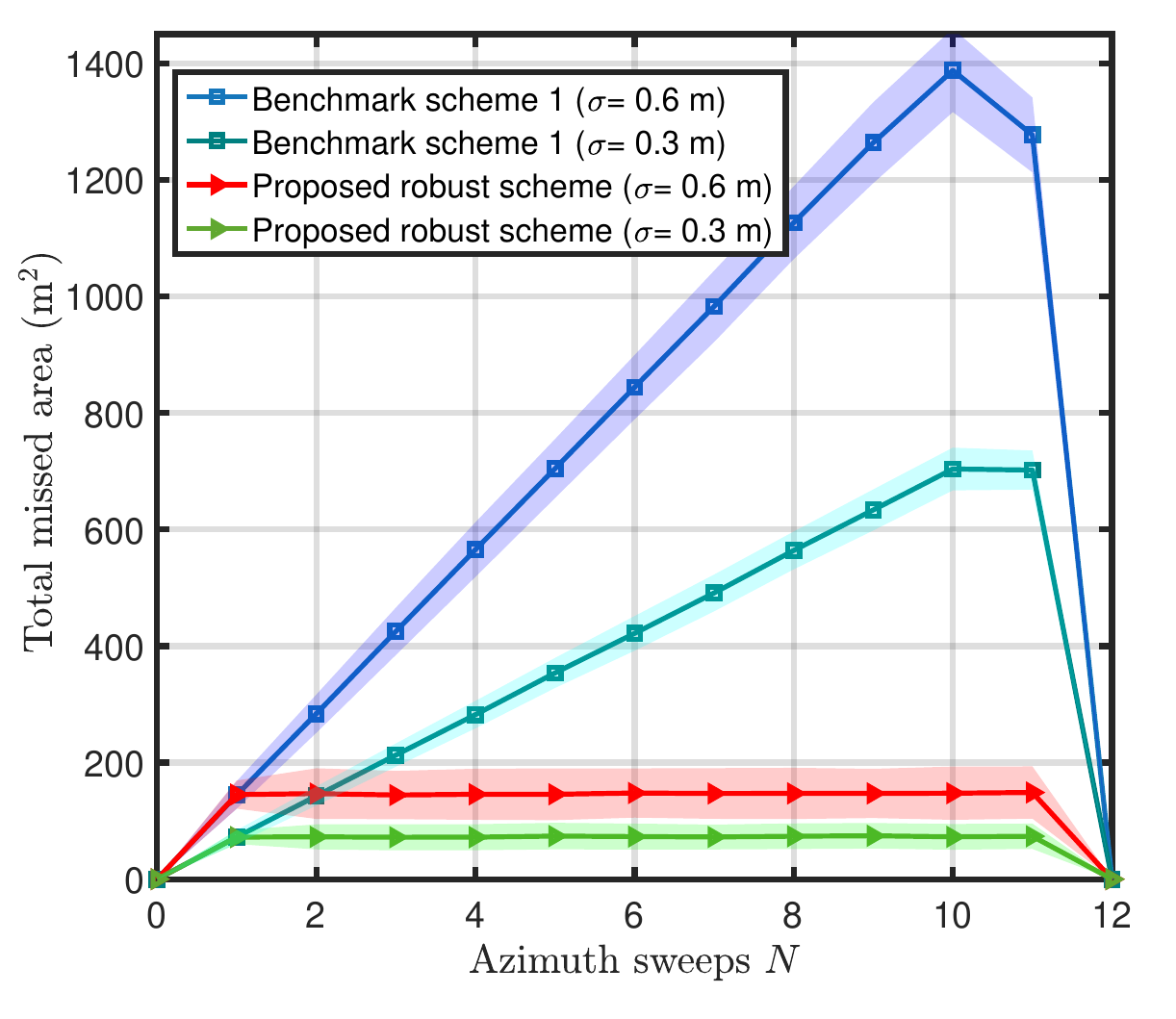}
\caption{Total missed area of the proposed robust scheme compared to the non-robust scheme for different trajectory deviation levels. Shaded area represents one standard deviation while solid line represents the mean value averaged over $10^4$ realizations.}
	\label{fig:missed} 
\end{figure}
\fi
In Figure \ref{fig:Convergence}, we analyze the convergence rate of the proposed sub-optimal and optimal (upper bound) schemes for different numbers of azimuth sweeps $N$. We also assess the gap between the optimal and sub-optimal solutions. The results are plotted as a function of the iteration number of \textbf{Algorithm} \ref{algorithm1} and \textbf{Algorithm} \ref{algorithm2}, respectively.   Figure \ref{fig:Convergence} shows that the proposed scheme converges fast to the upper bound on the optimal solution given by $C(z^*) \approx 9.154 \times 10^3$ m$^2$ for $N=1$, and $C(z^*) \approx 2.496\times10^4$ m$^2$ for $N=3$. After convergence, the gap between the upper bound on the optimal solution and the sub-optimal solution is negligible which confirms that the proposed sub-optimal scheme achieves close-to-optimal performance with low  complexity.\par 
{ In Figure \ref{fig:schemecomparison1},  we depict the sensing coverage as  a function of the number of azimuth sweeps $N$. The figure illustrates how the optimal $N^*$, leading to the best possible coverage, can be obtained based on a finite exhaustive search.} Furthermore, we compare the proposed scheme with the benchmark schemes.   In fact, the figure illustrates the superior performance of the proposed robust scheme over robust benchmark schemes 2 and 3. The gain is attributed to the optimization of the communication power and radar transmit power, respectively.  In comparison to non-robust benchmark scheme 1, our proposed robust scheme achieves comparable or improved coverage depending on the system parameters. Yet, the advantage of the proposed robust scheme over the non-robust scheme is illustrated in Figure \ref{fig:missed}, which shows the total area  missed by the \ac{uav} (i.e., the total area of all holes in the total coverage) averaged over $10^4$ realizations and for different levels of random trajectory deviations.  For the non-robust design, the missed area increases with the number of total azimuth sweeps, as more area is missed by the \ac{uav}-\ac{sar} due to random trajectory deviations. However, the proposed robust design keeps the missed area small and constant for different numbers of azimuth sweeps to ensure that the required $r=95\%$ coverage reliability is achieved. This confirms that, in terms of total missed area, the robust scheme outperforms the non-robust scheme, as the latter may cause significant coverage holes and thus, is not suitable in practice.
\subsection{Robustness of the Proposed Design Against \ac{uav} Trajectory Deviations} 

\begin{figure*}[] 
	\centering
	\begin{tabular}{ccccc}
c		\includegraphics[width=0.23\linewidth]{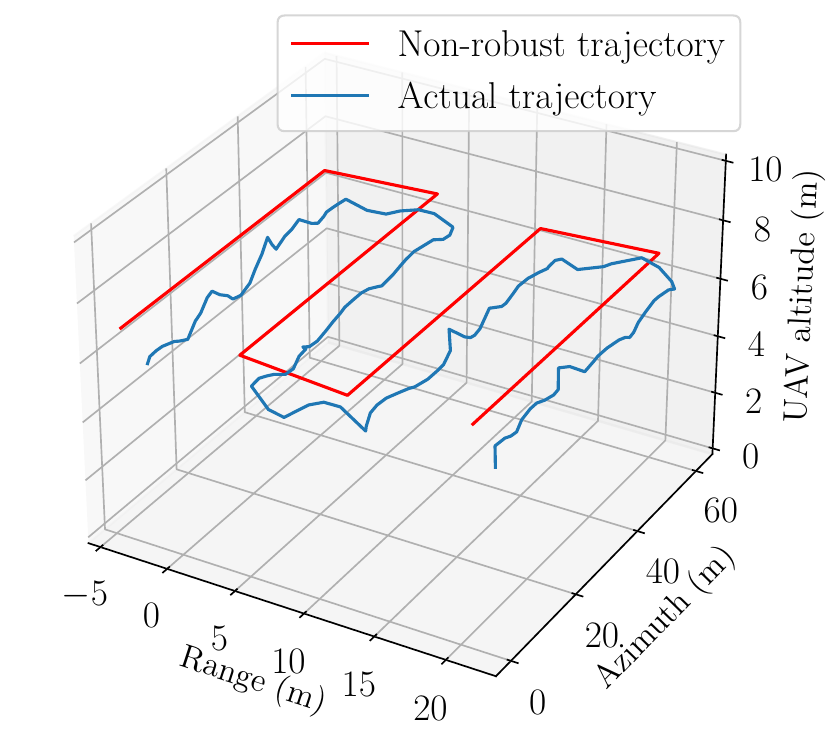}  &
		\includegraphics[width=0.23\linewidth]{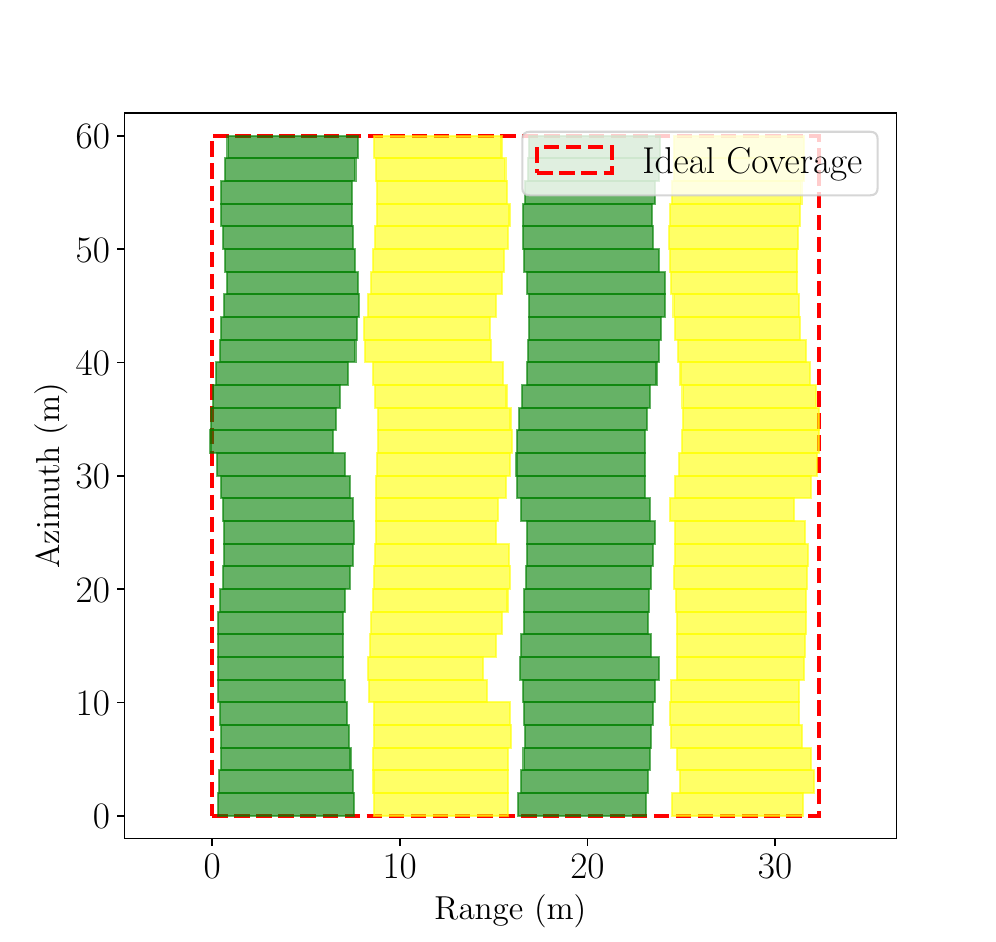}&
		\includegraphics[width=0.23\linewidth]{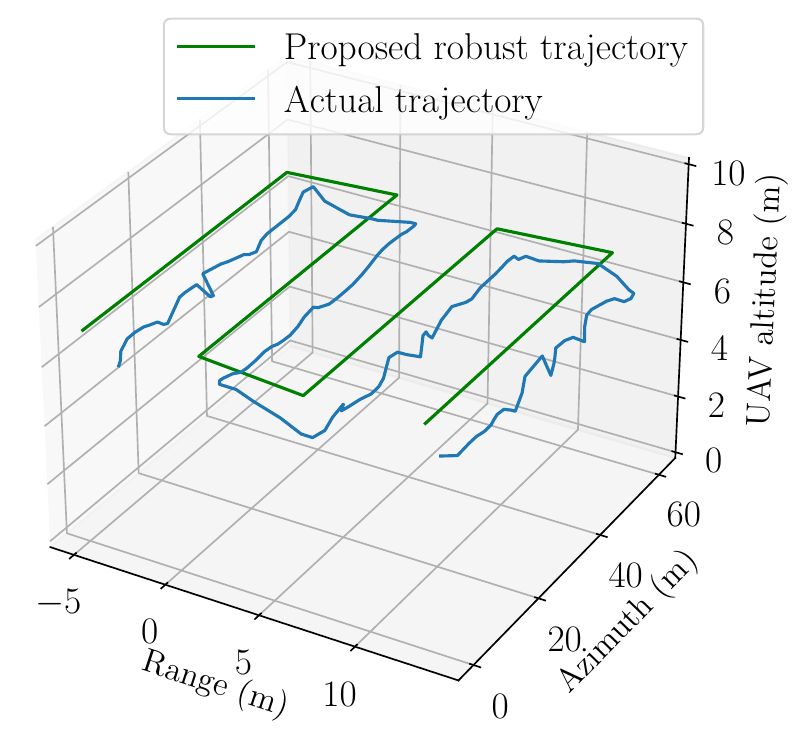}&
		\includegraphics[width=0.23\linewidth]{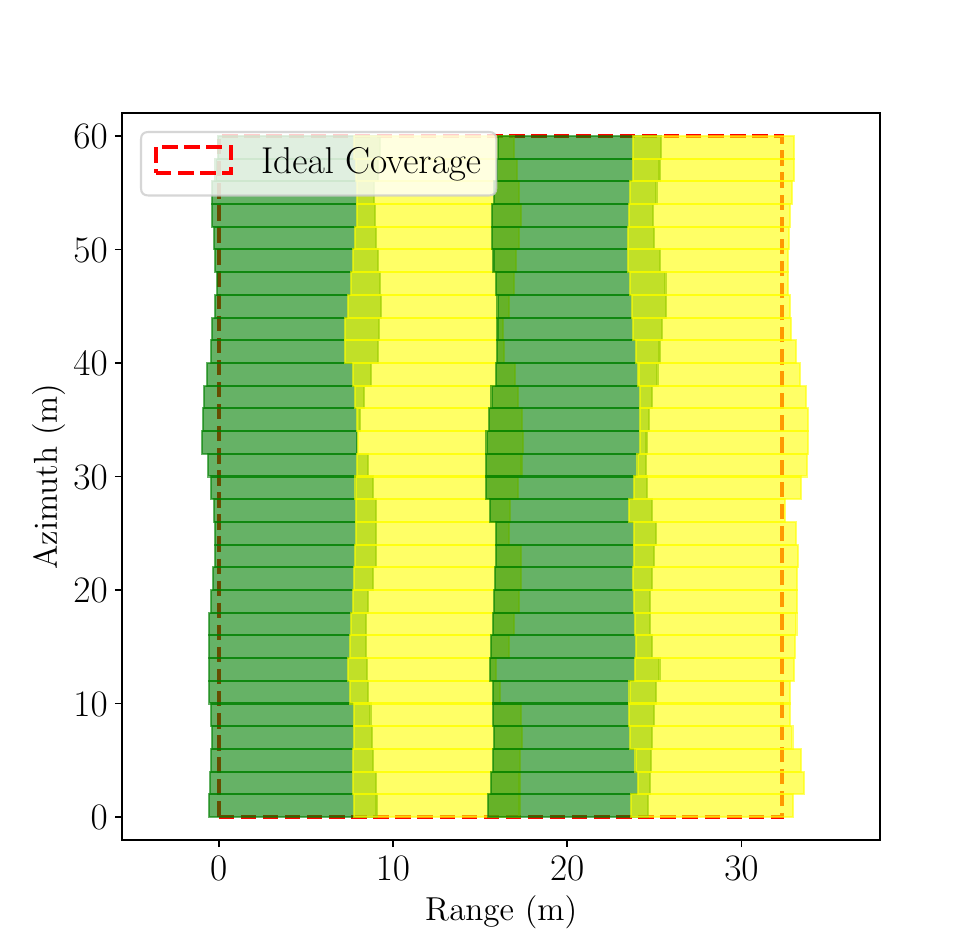} \\ 
		\textbf{(a)} & \textbf{(b)} & \textbf{(c)} & \textbf{(d)} \\
	\end{tabular}
	\caption{ Robust versus non-robust \ac{uav} trajectory design. (a) Non-robust \ac{3d} \ac{uav} trajectory for optimized coverage. (b) Actual \ac{2d} ground \ac{sar} coverage following the trajectory shown in (a). (c) Robust \ac{uav} trajectory for optimized coverage. (d) Actual \ac{2d} ground \ac{sar} coverage following the trajectory shown in (c). }
	\label{fig:trajectory}
\end{figure*} 

In Figure \ref{fig:trajectory}, we investigate the effect of \ac{uav} trajectory deviations on \ac{sar} mapping for low-altitude radar applications (i.e., $z_{\rm max}=10$ m).   In Figure  \ref{fig:trajectory}a, the dashed red line denotes the non-robust \ac{uav} trajectory provided by benchmark scheme 1, whereas the solid blue line  shows the actual \ac{uav} flight path. In Figure  \ref{fig:trajectory}b, we show the resulting \ac{uav}-\ac{sar} ground coverage with green- and yellow-colored rectangles denoting successive azimuth scans. Following an ideal path planning, the expected covered area is measured by the dashed red rectangle. However, the actual covered area is different from the expected result.  Figure \ref{fig:trajectory}b shows how small flight deviations can affect the \ac{uav} coverage, where the total uncovered area for benchmark scheme 1 was  $225$ m$^2$ out of  $C(\mathbf{z})=1.859\times 10^3$ m$^2$  (12\%).   In Figure  \ref{fig:trajectory}c, we present the  proposed robust \ac{uav} trajectory with a solid green line. Figure \ref{fig:trajectory}d highlights the tradeoff between coverage performance and robustness for low-altitude sensing applications. Although the total coverage is comparable, the proposed robust scheme is highly effective in avoiding uncovered areas caused by random \ac{uav} trajectory deviations,  as any location in the \ac{aoi} is guaranteed to be covered with a 95\% probability.

\ifonecolumn
\begin{figure}
	\centering
	\includegraphics[width=3.5in]{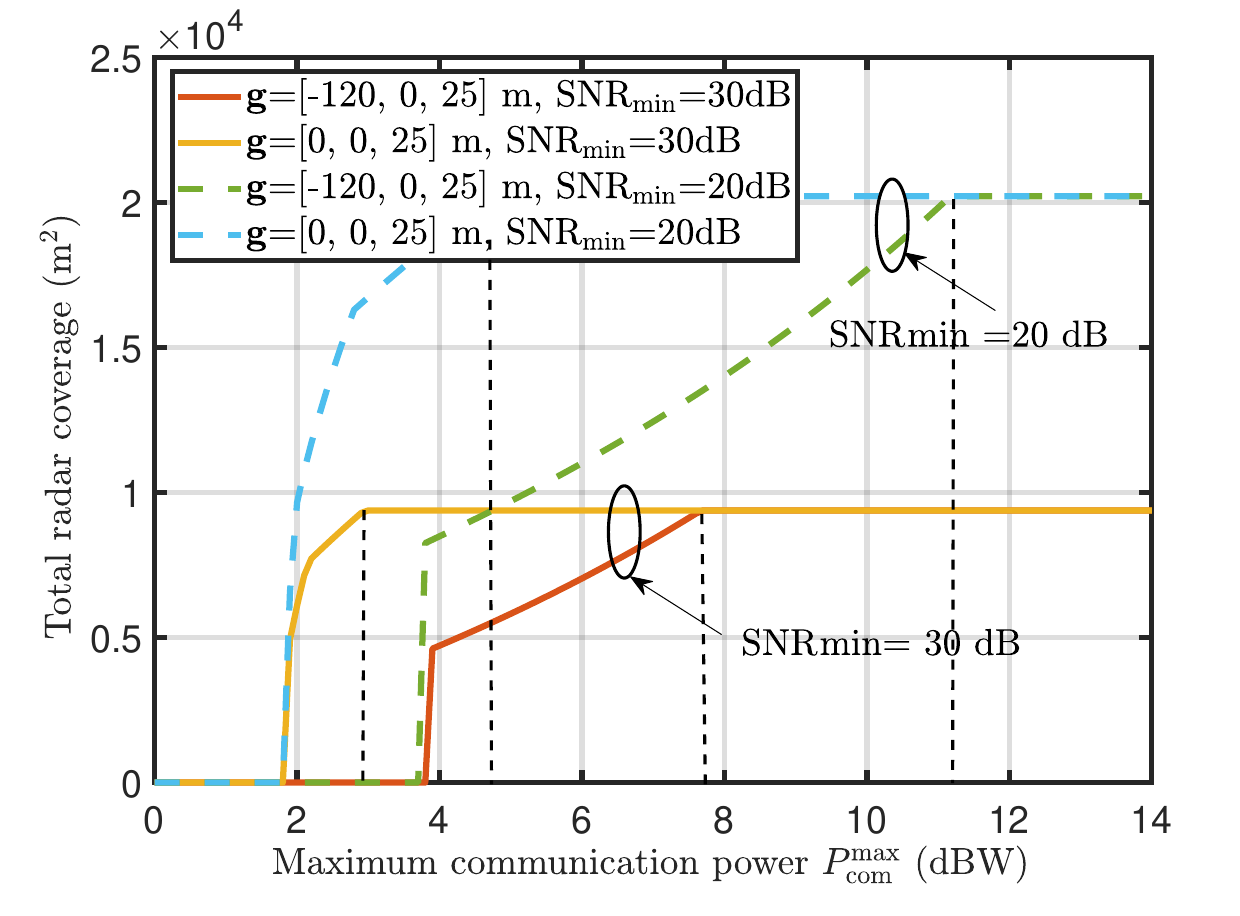}
	\caption{Total radar coverage versus maximum communication power.}
	\label{fig:communication}
\end{figure}
\fi
\subsection{Real-time Communication }
{In Figure \ref{fig:communication}, we plot the total radar ground coverage as a function of the \ac{uav}'s maximum communication power for different placements of the \ac{gs} and different required radar  $\mathrm{SNR}$ values.  This figure highlights the impact of communication on the sensing performance. In particular,  the communication transmit power affects the \ac{uav} trajectory (see (\ref{eq:data_rate})), and therefore also the sensing performance, as the trajectory  of the drone is linked to the ground coverage and the sensing \ac{snr} through (\ref{eq:covergeMetric}) and (\ref{equation:SNR}), respectively. Figure  \ref{fig:communication} shows that, when the \ac{gs} is far from the \ac{aoi}, i.e., for $\mathbf{g}=[-120,0,25]$ m, the total coverage is reduced as the minimum required data rate cannot be achieved at distances far from the \ac{gs} due to the increased path loss.} { This can be compensated by increasing the maximum \ac{uav} transmit power. In fact, for sufficiently high maximum \ac{uav} transmit powers,  the radar \ac{snr} constraint becomes the performance bottleneck causing the total coverage to saturate.} For instance, as shown in Figure \ref{fig:communication}, a maximum \ac{uav} communication power of $P_{\rm com}^{\rm max}=7.78 $ dBW results in the best possible radar coverage when $\rm SNR_{\rm min}= 30$ dB and $\mathbf{g}=[-120,0,25] $ m.  In a nutshell, a reliable data backhaul link is crucial for enabling real-time on-ground processing of \ac{sar} data. 
\subsection{ Impact of UAV Battery Capacity on Sensing Performance}
\ifonecolumn
\begin{figure}
	\centering
	\includegraphics[width=3.5in]{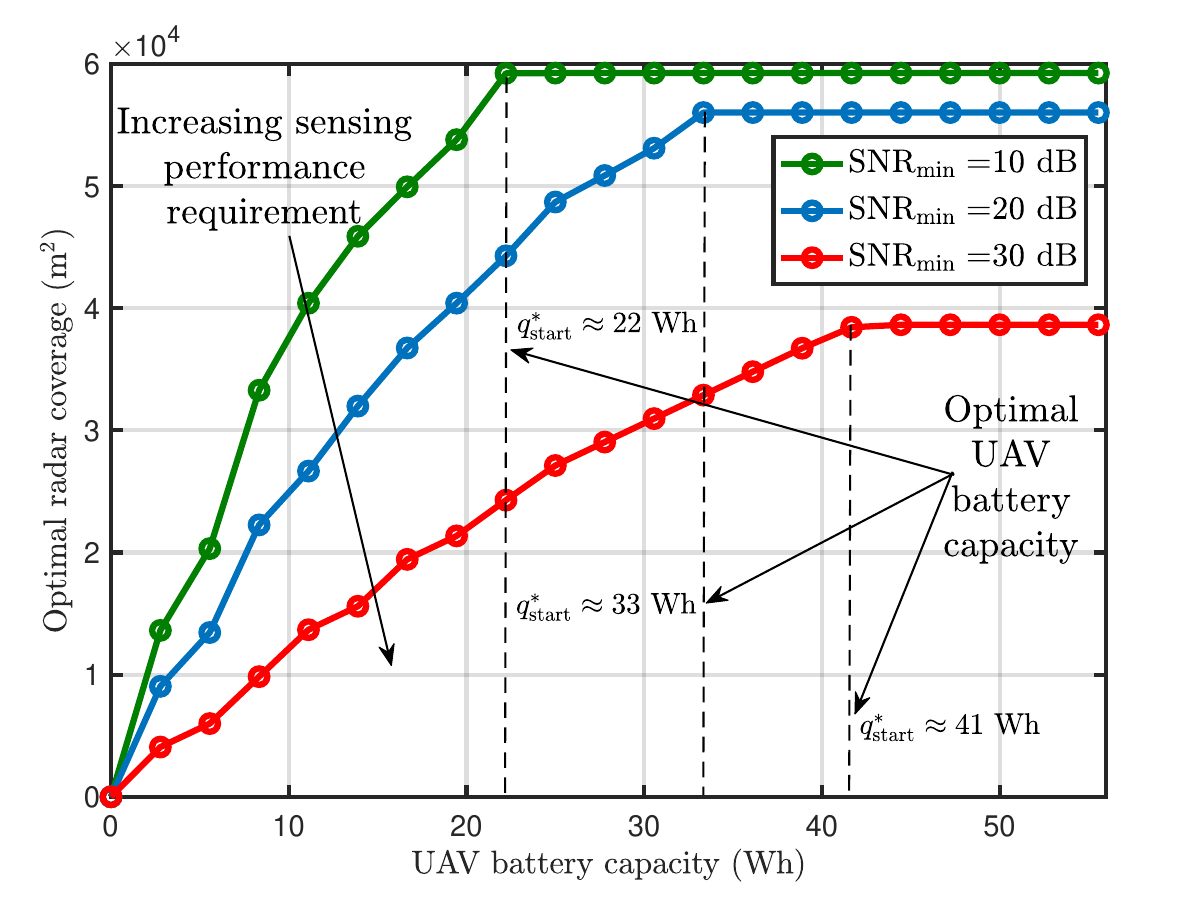}
	\caption{ Total radar coverage versus UAV battery capacity for different sensing requirements.}
	\label{fig:uav_battery}
\end{figure}
\fi
{ In Figure \ref{fig:uav_battery}, we depict the  \ac{sar} coverage as a function of the \ac{uav} battery capacity for different required sensing \acp{snr}, $\mathrm{SNR}_{\mathrm{min}}$. The figure shows that the optimum \ac{sar} coverage increases with the battery capacity and saturates for a certain value denoted by $q_{\mathrm{start}}^*$, which is the minimum \ac{uav} battery capacity needed to achieve maximum coverage for a given required sensing \ac{snr}. The coverage can not be further improved by increasing $q_{\mathrm{start}}$ beyond $q_{\mathrm{start}}^*$ due to communication constraint $\mathrm{C8}$, which eventually becomes the performance bottleneck for large \ac{uav}-\ac{gs} distances. In addition, Figure \ref{fig:uav_battery} reveals that for more stringent sensing requirements, not only does the total coverage decrease, but the energy needed to achieve maximum coverage also increases. This is  because the flying altitude has to be reduced to satisfy more stringent \ac{snr} requirements leading to less coverage. However, when the \ac{uav} is forced to fly at lower altitudes, a larger number of sweeps have to be performed to maximize the coverage, which results in a higher energy consumption.  For example, for $\mathrm{SNR}_{\mathrm{min}}$ = 10 dB, $C(z^*) = 5.9 \times 10^4$ m$^2$ is obtained with a 22 Wh battery capacity, whereas for $\mathrm{SNR}_{\mathrm{min}}$ = 30 dB, only $C(z^*) = 3.8 \times 10^4$ m$^2$ is achieved requiring almost twice the battery capacity, i.e., 41 Wh. This further shows that the limited energy available to small \acp{uav}, which is accounted for by by constraints $\mathrm{C9-C11}$, has a significant impact on sensing performance.} 

\section{Conclusion}\label{Sec:Conclusion}
In this paper, we investigated robust joint \ac{3d} trajectory and resource allocation design for \ac{uav}-\ac{sar} systems, where real-time transmission of radar data to the ground was taken into account. Random \ac{uav} trajectory deviations were modeled and their effect on the ground radar coverage was derived based on their statistics. Moreover, the robust \ac{3d} trajectory and resource allocation design was formulated as a mixed-integer non-linear optimization problem for maximization of the \ac{uav} radar coverage. A low-complexity sub-optimal solution based on \ac{sca} was provided. To assess the performance of the proposed algorithm, we derived an upper bound for the optimal solution based on monotonic optimization theory. Simulation results confirmed that the proposed robust sub-optimal solution achieves close-to-optimal performance. The superior performance of the proposed scheme compared to three benchmark schemes was established and the robustness of the proposed solution \ac{wrt} \ac{uav} trajectory deviations was confirmed. Finally, we highlighted the impact of \ac{gs} placement and the  power available for communication on \ac{uav} radar performance. { In this paper, we assumed a deterministic \ac{los} link between the \ac{uav} and the \ac{gs}. An interesting topic for future work is the extension to other types of air-to-ground channels, including Ricean fading channels, based on online optimization techniques, see, e.g., \cite{complexity1}. Furthermore, while we used classical optimization methods and focused on single-\ac{uav} \ac{sar} sensing, the extension to multi-\ac{uav} \ac{sar} sensing \cite{conf2} and the exploration of  learning-based methods for the trajectory design and resource allocation are promising research directions.}
\appendices

\section{Proof of Proposition \ref{proposition:Near-Far-Range}}
\label{app:P.1}
Assuming Gaussian \ac{uav} position deviations, i.e., $\Delta_x(n)\sim\mathcal{N}( o_x,\sigma)$ and $\Delta_z(n)\sim\mathcal{N}(o_z,\sigma)$, {and} using (\ref{eq:Delta_N}) and (\ref{eq:Delta_F}),  we can show that $\Delta_{\mathbf{p}_N}(n)\sim \mathcal{N}\big(o_x+c_1 o_z,\sigma\sqrt{1+c_1^2}\big)$ and $\Delta_{\mathbf{p}_F}(n)\sim \mathcal{N}\big(o_x+c_2 o_z,\sigma\sqrt{1+c_2^2}\big), \forall n \in \mathbb{N}_{NM}$.  In time slot $n$, let set $\mathcal{I}_N=\{ s \in \mathbb{R} \hspace{2mm}\big | \hspace{2mm}\mathbb{P}\left(\Delta_{\mathbf{p}_N}(n)+ s\leq 0 \right) \geq r   \}$, where $r$ is the coverage reliability. Based on the definition of the near-range and far-range compensation provided in (\ref{eq:Near-rangeDefinition}) and (\ref{eq:Far-rangeDefinition}), respectively, we can write:
\begin{align}
	F_{\Delta_{\mathbf{p}_N}(n)}(-s) &\geq r,\\
	1 -F_{-\Delta_{\mathbf{p}_N}(n)}(s) &\geq r,\\
	\frac{1}{2}\left(1 -\mathrm{erf}\left( \frac{s+ o_x+c_1 o_z}{\sigma \sqrt{2(1+c_1^2})}\right)\right)&\geq r,\\
	\mathrm{erf}\left( \frac{s+ o_x+c_1 o_z}{\sigma \sqrt{2(1+c_1^2})}\right) &\leq 1-2r.
\end{align}
Since the inverse of the error function, denoted by $\rm erf^{-1}(\cdot)$, is an increasing function, we have: 
\begin{gather}
	s\leq  h_{\mathbf{p}_N}(r),
\end{gather}
where $h_{\mathbf{p}_N}(r)=\mathrm{erf}^{-1}( 1-2r)\sigma \sqrt{2(1+c_1^2})-o_x-c_1 o_z$.
Then, the required near-range compensation $\delta_{\mathbf{p}_N}^{r}(n)=\argmin(|s|)$ is given by:
\begin{gather}
	\delta_{\mathbf{p}_N}^{r}(n)=\begin{cases}
		0  & \mathrm{if } \hspace{2mm}h_{\mathbf{p}_N}(r)\geq 0,  \\
		h_{\mathbf{p}_N}(r) & \mathrm{if } \hspace{2mm}h_{\mathbf{p}_N}(r)\leq 0,
	\end{cases}
\end{gather}
which can be equivalently written as follows:
\begin{gather}
	\delta_{\mathbf{p}_N}^{r}(n)=\min\big(0,h_{\mathbf{p}_N}(r)\big)=- [-h_{\mathbf{p}_N}(r)]^+.
\end{gather}
This gives the final result in (\ref{eq:deltaN}). The far-range compensation, denoted by $\delta_{\mathbf{p}_N}^{r}(n)$, is derived using the same methodology.

\section{Proof of Proposition \ref{prop:feasible}}
\label{ap:proposition3}
	To show that the set of feasible  azimuth scans is finite, we prove that $N$ is upper bounded. From constraints $\rm C9$ and  $\rm C11$, we obtain the following inequality: 
	\begin{equation}
		\delta_t P_{\mathrm{prop}} \leq q(n) -q(n+1){, \forall n \in \mathbb{N}_{NM-1}.} \label{eq:lowerboundLHS}
	\end{equation}
	The summation of both sides of (\ref{eq:lowerboundLHS}) yields:
	\begin{equation}
		\sum_{n=1}^{NM-1}\delta_t P_{\mathrm{prop}} \leq q(1)-q(NM) \leq q_{\rm start}.\label{eq:telescop}
	\end{equation}
	The result in (\ref{eq:telescop}) can be rigorously proven using induction applied to the telescoping series $q(n)-q(n+1)$ and constraint $\rm C10$.
	Since the drone flies at a constant velocity $v$ and $P_{\rm prop}$ is constant, we obtain an upper bound on $N$ as follows: 
	\begin{equation}
		N \leq \frac{1}{M} \left(\frac{q_{\mathrm{start}}}{\delta_t P_{\mathrm{prop}}}+1\right). \label{eq:upperboundN}
	\end{equation}
	Problem $\rm (P.1)$ is infeasible for $N$ exceeding the upper bound provided in (\ref{eq:upperboundN}) as this would violate constraint $\rm C11$.

 \section{Proof of Theorem \ref{theo:P.2}} 
	\label{app:P.2}
	In this proof, we show that problem $\rm (\overline{P.2})$ is an upper bound to problem $\rm (P.2)$.   To this end, we introduce an intermediate problem, denoted by $\rm (P.2')$, that provides an upper bound for problem $\rm (P.2)$, then we prove that problems $\rm (P.2')$  and  $\rm (\overline{P.2})$ are equivalent. This means that $\rm (\overline{P.2})$ is an upper bound to the original problem  $\rm (P.2)$. We start by rewriting constraint $\rm C8$ as follows: 
	\begin{align}  \label{eq:theoC8}
		{\rm C8}: g_1(n) \leq -g_2(n) (y(n)-g_y)^2, \forall n \in \mathbb{N}_{NM},
	\end{align}
	where functions $g_1$ and $g_2$ are respectively given by: 
	\ifonecolumn
		\begin{align} 
		g_1(n )&= g_2(n) \Big( (x^{r}(n)-g_x)^2 + (z^{r}(n)-g_z)^2\Big) - P_{\mathrm{com}}(n)\gamma,\\
		g_2(n )&=  A2^{\alpha z^{r}(n)} -1, \forall n \in \mathbb{N}_{NM}.
	\end{align}
	\fi
	As proposed in (\ref{eq:lowerboundy}), we lower bound the term $(y(n)-g_y)^2$ in (\ref{eq:theoC8}) by taking its minimum value for each azimuth scan. Since $g_2(n)\geq 0,  \forall n \in \mathbb{N}_{NM}$, by bounding $(y(n)-g_y)^2$, the right-hand side value in (\ref{eq:theoC8}) increases.   This is equivalent to relaxing constraint $\rm C8$, which we denote now by $\rm \overline{C8}$: 
	\ifonecolumn
		\begin{gather}  
		{\rm \overline{C8}}: g_1(n) \leq -g_2(n) \min\limits_{1+(n-1)M\leq k \leq nM}(y(k)-g_y)^2, \forall n \in \mathbb{N}_{N}.
	\end{gather}
	\fi
	Let us denote problem $\rm (P.2)$ after relaxing constraint $\rm C8$ by $\rm (P.2')$. It is easy to show that $\rm (P.2')$ provides an upper bound to problem $\rm (P.2)$. In fact, any feasible solution to $\rm (P.2)$ is also feasible for $\rm (P.2')$ as the solution necessarily satisfies ${\rm \overline{C8}}$. Now, let us denote the optimal altitudes to problem $\rm (P.2)$ and $\rm (P.2')$ by $\mathbf{z} \in \mathbb{R}^{NM \times 1}$ and $\mathbf{z}_2 \in \mathbb{R}^{NM \times 1}$,   respectively, and their robust trajectories by  $\mathbf{u}^r$ and $\mathbf{u}_2^r$,   respectively. Since the objective function is increasing \ac{wrt} altitude vector $\mathbf{z}$, then, relaxing $\rm C8$ results in the following inequality:
	\begin{equation}
		 C(\mathbf{u}^r) \leq C(\mathbf{u}^r_2), \forall r \in [0,1]. \label{eq:majorization1}
	\end{equation} In the remainder of this proof, we show that problems $\rm (P.2')$ and $(\overline{\mathrm{P.2}})$ are equivalent.\par
	We apply a series of transformations to transform $NM$-dimensional problem $\rm (P.2')$ to the $N$-dimensional problem $(\overline{\mathrm{P.2}})$. In fact, based on $\mathrm{\overline{C8}}$, where the \ac{uav}-\ac{gs} distance  varies from scan to scan, we remark that all optimization variables of problem $\rm (P.2')$  are also constant for a given azimuth scan and only change from scan to scan. This allows us to reduce the dimension of the problem as optimization along the azimuth direction has no impact. Therefore, constraints $\rm C3-C5$   are dropped as they define equalities for time slots of the same azimuth scans. The rest of the constraints are now specified for each scan and denoted by  $\overline{\rm C6} -\overline{\rm C11}$. We use $N$-dimensional optimization variables,  where based on ${\rm \overline{C8}}$ and constraints $\rm C1$ and $\rm C2$, optimization vector $\overline{\mathbf{x}}\in \mathbb{R}^{N \times1}$ can be dropped from the optimization problem $(\mathrm{P.2'})$ as it is deterministic and can be calculated from vector $\overline{\mathbf{z}}\in \mathbb{R}^{N\times1}$ as follows:
	\begin{align} \label{eq:lemma2}
		\overline{x}(n)=(c_2-c_1)\sum_{k=1}^{n-1}  \overline{z}(k)-c_1 \overline{z}(n),  n \in \mathbb{N}_{N}.
	\end{align}
 Finally, we obtain the following relation between the objective functions of problems $(\mathrm{P.2'})$ and  $(\overline{\mathrm{P.2}})$:
	\ifonecolumn
		\begin{align}
C(\mathbf{u}_2^{r})&= \sum_{n=1}^{NM}\Delta_s  \left(z_2(n)+\delta^{r}_z(n)\right) \left(c_2-c_1\right)\\ &\stackrel{\mathrm{C4}}{=}L\sum_{k=1}^{N}  \Big(z_2(1+(1-k)M)+\delta^{r}_z(1+(1-k)M)\Big) \left(c_2-c_1\right)\\ 
&= \frac{L}{\Delta_s}\sum_{n=1}^{N} \Delta_s \overline{z^r}(n) \left(c_2-c_1\right)=\frac{L}{\Delta_s} C(\overline{\mathbf{u}^{r}}). \label{eq:majorization2}
	\end{align}
	\fi
	Based on (\ref{eq:majorization1}) and (\ref{eq:majorization2}), problem $\rm (P.2')$ is equivalently transformed to problem  $\rm (\overline{P.2})$, with $C(\mathbf{u}^{r}) \leq	C(\mathbf{u}_2^{r})=\frac{L}{\Delta_s} C(\overline{\mathbf{u}^{r}})$, which concludes the proof. 

\bibliographystyle{IEEEtran}
\bibliography{biblio}


\end{document}